\documentclass[11pt]{article}
\usepackage{fullpage}
\usepackage[bookmarks]{hyperref}
\usepackage{amssymb}
\usepackage{amsmath}
\usepackage{amsthm}
\usepackage{youngtab}
\usepackage{graphicx,color,colordvi}
\usepackage{subfig}
\usepackage{bbm}
\usepackage{cleveref}
\usepackage[utf8]{inputenc}
\usepackage{fullpage}
\usepackage[blocks]{authblk}
\usepackage{mathtools}
\usepackage{multirow}
\usepackage{dsfont}
\usepackage[toc,page]{appendix}

\newcommand{\nc}{\newcommand}
\nc{\rnc}{\renewcommand}
\nc\mnb[1]{\medskip\noindent{\bf #1}}

\newcommand{\ket}[1]{\left| #1 \right>} 
\newcommand{\bra}[1]{\left< #1 \right|} 

\newcommand{\cC}{{\mathbb{C}}}
\newcommand{\cS}{\mathcal{S}}
\newcommand{\M}{{\mathbb{M}}}
\newcommand{\cV}{{\rm{vec}}}
\newcommand{\mat}{{\rm{mat}}}
\newcommand{\tPi}{{\widetilde{\Pi}}}
\newcommand{\tGamma}{{\widetilde{\Gamma}}}
\newcommand{\id}{{\rm{id}}}

\newcommand{\ot}{\otimes}

\newcommand{\End}{\operatorname{End}}
\newcommand{\Ad}{\operatorname{Ad}}
\newcommand{\Int}{\operatorname{Int}}
\newcommand{\sgn}{\operatorname{sgn}}

\newcommand{\<}{\langle}
\newcommand{\cH}{\mathcal{H}}
\newcommand{\cK}{\mathcal{K}}
\renewcommand{\>}{\rangle}
\newcommand\be{\begin{equation}}
\newcommand\ee{\end{equation}}

\DeclareMathOperator{\tr}{Tr}

\usepackage{hyperref}
\hypersetup{colorlinks, citecolor=magenta, filecolor=blue, linkcolor=blue, urlcolor=green}
\newtheorem{theorem}{Theorem}

\newtheorem{corollary}[theorem]{Corollary}

\newtheorem{definition}[theorem]{Definition}
\newtheorem{example}[theorem]{Example}

\newtheorem{lemma}[theorem]{Lemma}

\newtheorem{proposition}[theorem]{Proposition}
\newtheorem{remark}[theorem]{Remark}

\newtheorem{assum}{Assumption}
\begin{document}
 	
	\title{Structure of irreducibly covariant quantum channels for finite groups }

  	\author{Marek Mozrzymas}
\affil[1]{\small Institute for Theoretical Physics, University of Wrocław
50-204 Wrocław, Poland} 
\author{Micha{\l} Studzi{\'n}ski}
\affil[2]{\small DAMTP, Centre for Mathematical Sciences, University of Cambridge, Cambridge~CB30WA, UK}
\author{Nilanjana Datta}
 	\affil[3]{\small Statistical Laboratory, Centre for Mathematical Sciences, University of Cambridge, Cambridge~CB30WB, UK}

 	\maketitle			 
\begin{abstract}
We obtain an explicit characterization of linear maps, in particular, quantum channels, which are covariant with respect to an irreducible representation ($U$) of a finite group ($G$), whenever $U \otimes U^c$ is simply reducible (with $U^c$ being the contragradient representation). Using the theory of group representations, we obtain the spectral decomposition of any such linear map. The eigenvalues and orthogonal projections arising in this decomposition are expressed entirely in terms of representation characteristics of the group $G$. This in turn yields necessary and sufficient conditions on the eigenvalues of any such linear map for it to be a quantum channel. We also obtain a wide class of quantum channels which are irreducibly covariant by construction. For two-dimensional irrreducible representations of the symmetric group $S(3)$, and the quaternion group $Q$, we also characterize quantum channels which are both irreducibly covariant and entanglement breaking. 
\end{abstract}

\section{Introduction}
Quantum channels are fundamental building blocks of any quantum communication- or information-processing system. A channel (classical or quantum) is inherently noisy, and its potential for communication is quantified by its capacity, i.e.~the maximum rate at which information can be transmitted reliably through it. Unlike its classical counterpart, a quantum channel has a number of different capacities. These depend on various factors, e.g.~the nature of information transmitted (classical, private classical or quantum), the nature of the input states (entangled or product), the absence or presence of any additional resource, e.g.~prior shared entanglement between the sender and the receiver, the nature of the measurements, if any, done on the output states (collective or individual) etc. Evaluating the capacities of quantum channels is one of the most important and challenging problems in quantum information theory. The problem becomes more tractable, however, if the channel in question satisfies certain symmetries. 
Suitably exploiting these symmetries leads to a simplification of the problem and allows us to infer more about the properties of the channel.

To elucidate this, let us consider the simplest communication scenario, namely, the transmission of classical information through a memoryless quantum channel (say, $\Phi$). The latter is a channel for which each use is independent and identical to all prior uses. Hence 
there is no correlation in the noise acting on the inputs to successive uses of the channel. The classical capacity, $C(\Phi)$, is the maximum number of bits that can be reliably transmitted per use of the channel. It is evaluated in the limit $n \to \infty$ (where $n$ denotes the number of uses of $\Phi$), under the constraint that the error incurred in the protocol vanishes in this limit. Thanks to the celebrated Holevo-Schumacher-Westmoreland theorem \cite{Holevo},\cite{SW97}, the classical capacity $C(\Phi)$, of any memoryless quantum channel $\Phi$ is known to be given by  the following {\em{regularized}} expression:
\begin{align}\label{c-cap}
C(\Phi) &= \lim_{n \to \infty} \frac{1}{n} \chi^*(\Phi^{\otimes n}),
\end{align}
where $\chi^*(\Phi)$ is called the {\em{Holevo capacity}} of the channel. The latter is equal to the classical capacity of $\Phi$ evaluated 
under the constraint that the input states are product states.
It is given by the following single-letter expression:
\begin{align}\label{Hol-cap}
\chi^*(\Phi)&:= \sup_{\{p_i, \rho_i\}} \left\{S\left(\sum_i p_i \Phi(\rho_i) \right) - \sum_i p_i S(\Phi(\rho_i))\right\},
\end{align}
where the supremum\footnote{If the output Hilbert space ($\cK$, say) of the channel is finite
	dimensional, then the supremum in \cref{Hol-cap} is a maximum, and is attained for an ensemble of
	at most $({\rm{dim}} \cK)^2$
	states} is taken over ensembles of quantum states $\rho_i$ occurring with probabilities $p_i$, and $S(\rho):= - \tr (\rho \log \rho)$ denotes the von Neumann entropy of a state $\rho$.

Unfortunately, due to the regularization present in it, the expression (\ref{c-cap}) of the (unconstrained) classical 
capacity, $C(\Phi)$, is in general intractable, and cannot be computed. 
If, however, the Holevo capacity of the channel is additive, then the expression for the classical capacity reduces to a 
single letter one: $C(\Phi)= \chi^*(\Phi)$, and in this case, it is clear that using entangled inputs does not provide any advantage in the transmission of classical information through $\Phi$. By providing a counterexample to the so-called {\em{additivity conjecture}}, which
had been the focus of active research for more than a decade, Hastings \cite{Has09} proved that the Holevo capacity of a quantum channel is not necessarily additive.

Determining whether the Holevo capacity of a given channel is additive or not is a rather non-trivial
problem. However, there are important families of channels satisfying certain symmetry properties, for which this problem can be simplified. 
These are the so-called {\em{irreducibly covariant}} channels, defined below. If $\Phi$ is such a channel, then its Holevo capacity, $\chi^*(\Phi)$, is linearly related to its minimum output entropy 
$S_{\min}(\Phi) := \min_\rho S(\Phi(\rho))$. Hence the problem of determining whether $\chi^*(\Phi)$ is additive reduces to the relatively simpler problem
of determining whether $S_{\min}(\Phi)$ is additive. The additivity of the Holevo capacity has been successfully established
for various families of irreducibly covariant channels (see e.g.~\cite{irr-cov1, irr-cov3, irr-cov2, Fan1, irr-cov4} and references therein), by proving that the corresponding minimum output entropy is additive.

Let us introduce the definition of irreducibly covariant channels. Let $G$ be a finite (or compact) group and for every $g \in G$, let $U(g)$ and $V(g)$ be unitary representations in the Hilbert spaces 
$\cH$ and $\cK$ respectively. Then a quantum channel $\Phi$, with input Hilbert space $\cH$ and output Hilbert space $\cK$, is said to be {\em{covariant}} with respect to these representations if for any input state $\rho$, 
\begin{align}
\Phi\left( U(g) \rho U(g)^\dagger\right) = V(g) \Phi(\rho) V(g)^\dagger \quad \forall \,\, g \in G. 
\end{align}
Moreover, if the representations $U$ and $V$ are irreducible, then the channel is said to be {\em{irreducibly covariant}}. For an irreducibly covariant channel
$\Phi$, the Holevo capacity and minimum output entropy satisfy the following linear relation~\cite{Holevo02}: 
$\chi^*(\Phi) = \log d - S_{\min}(\Phi),$ where $d:= {\rm{dim}}\, \cK$.

The symmetries underlying (irreducibly) covariant channels make them amenable to analysis in various information-theoretic problems, e.g.~establishing strong converse properties of classical capacities, channel discrimination, obtaining second order asymptotics for 
entanglement assisted- and private classical communication etc. K{\"o}enig and Wehner \cite{KW09} proved that the classical capacity of any covariant 
quantum channel for which the minimum output entropy is additive, satisfies the so-called {\em{strong converse property}}. 
That is, any communication protocol with rate larger than the classical capacity fails with certainty (i.e.~the probability of error 
in transmission converges to one) in the limit of asymptotically many uses of the channel. In \cite{DTW16}, the strong converse property 
of the entanglement-assisted classical capacity of covariant channels was established. Recently, Jencova proved \cite{JP16} that 
two irreducibly covariant channels can be discriminated using an optimal scheme employing a maximally entangled input state. Wilde et al \cite{WTB16} 
obtained second order expansions of relative entropy of entanglement bounds for private communication rates for covariant channels.
These bounds also hold in the standard setting of private communication through a quantum channel, in which the sender and
receiver have access to unlimited public classical communication. The bounds are useful for establishing converse bounds for quantum 
key distribution protocols conducted over these channels. Finally covariance with respect to some finite or compact group can we used to investigate asymmetry properties of pure quantum states~\cite{Mar1}. The above discussion and examples illustrate the relevance of (irreducibly) covariant 
channels, and highlights the importance of understanding the structure and properties of such channels. 

There have been some notable results related to the characterization of covariant channels, and a study of their properties in particular for some low dimensional compact groups. For example, Scutaru~\cite{Scutaru} proved a Stinespring type theorem, in the  $C^*$-algebraic framework, for any completely positive linear map 
which is covariant with respect to a unitary representation of a locally compact group. In~\cite{Nuw} covariance of quantum channels with respect to the $SU(2)$ group was investigated. In~\cite{Nuw} the notion of EPOSIC channels was introduced, and as an application some new positive maps, which are not completely positive, were derived. It is worth mentioning here, that imposing covariance with respect to $SU(2)$ gives a way to understand more about entanglement in quantum spin systems~\cite{Schliemann} and also allows us to prove an extended version of the Lieb-Mattis-Schultz theorem by use of Matrix Product States~\cite{Sanz1}. Moreover, taking the group $SU(n)$, we can obtain direct proof of dimerization of quantum spin chains~\cite{Nach1}. Mendl and Wolf~\cite{Mendl} studied unital channels which are covariant with respect to the real orthogonal group, and determined the subset of these channels which are convex combinations of unitaries. In~\cite{irr-cov1} complementarity and additivity of various covariant channels, such as the depolarizing and Weyl-covariant channels, were studied.  

In this paper we obtain a detailed mathematical description of channels which are irreducibly covariant with respect to a finite 
group $G$ in the case in which: $(i)$ the input and output Hilbert spaces of the channel are the same, and $(ii)$ if $U$ is the 
particular unitary irreducible representation (irrep) considered,
then $U \otimes U^c$ is simply reducible (or multiplicity free), where $U^c$ denotes the contragradient representation, i.e. for every $g\in G$, $U^c(g)=U\left(g^{-1}\right)^T \equiv {\overline{U}}(g)$. Firstly, we obtain the spectral decomposition
of the Choi-Jamio{\l}kowski image of any linear map which is irreducibly covariant. The eigenvalues and orthogonal projections arising 
in this decomposition are expressed entirely in terms of representation characteristics of the group. This in turn yields necessary and sufficient conditions on the eigenvalues of the linear map, for which it is a quantum channel (i.e.~a completely positive and trace-preserving map). We also obtain explicit expressions for the Kraus operators of such channels. See \Cref{thm16} and \Cref{KK1} of \Cref{sec:main}. Moreover we give geometrical interpretation of the set of the solutions for which given irreducibly covariant linear map (ICLM) is a irreducibly covariant quantum channel (ICQC). Namely we show, that all eigenvalues of the Choi-Jamio{\l}kowski image of an ICQC neccesarliy lie in the intersection of contracted simplex and certain subspace defined by the matrix obtained spectral analysis of the projectors appearing in the decomposition of an ICLM. See~\Cref{Mrank}, and~\Cref{intersection} of~\Cref{geometry}.

Using all above-mentioned characterization,  for {\em{any}} finite, multiplicity free group, we obtain a wide class of quantum channels which are irreducibly covariant by construction.
In addition, we provide explicit examples of ICQCs for certain multiplicity free groups, namely, the symmetric groups, $S(3)$ and $S(4)$, and the quaternion group $Q$. In each case we present both the matrix representation and the Kraus representation of the ICQC. Further, for the case of $S(3)$ and $Q$, using the Peres-Horodecki or positive partial transpose (PPT) criterion~\cite{SepHHH,SepAP}, we also obtain the condition under which the ICQC is an entanglement breaking channel (EB)~\cite{EBB}.

To give a flavour of our results, let us consider the example of the symmetric group $G=S(3)$, and the family of channels which are 
irreducibly covariant with respect to the two-dimensional irrep $U$ characterised by the partition $
\lambda =(2,1)$ using the so-called $\epsilon$-representation \cite{NS}.
We prove (see \Cref{ex-S(3)}) that the Kraus operators for this family of channels is given by:
\be
\label{expK}
\begin{split}
	K_1(\lambda)&=\sqrt{\frac{1}{2}(1-l_{\sgn})}\begin{pmatrix}
		0 & 0\\ 1 & 0
	\end{pmatrix},\quad K_2(\lambda)=\sqrt{\frac{1}{2}(1-l_{\sgn})}\begin{pmatrix}
	0 & 1\\0 & 0
\end{pmatrix},\\
K_3(\sgn)&=\sqrt{\frac{1}{2}(1+l_{\sgn}-2l_{\lambda})}\begin{pmatrix}
	-\frac{1}{\sqrt{2}} & 0\\ 0 & \frac{1}{\sqrt{2}}
\end{pmatrix},\quad 	K_4(\id)=\sqrt{\frac{1}{2}(1+l_{\sgn}+2l_{\lambda})}\begin{pmatrix}
\frac{1}{\sqrt{2}} & 0\\ 0 & \frac{1}{\sqrt{2}}\end{pmatrix},
\end{split}
\ee
where $\id,\sgn,\lambda$ denote inequivalent irreps of the symmetric group $S(3)$. In the~\cref{expK}, $l_{\sgn}$ and 
$l_\lambda$ are two real parameters, which are constrained to take values in the triangle shown in the left panel of \Cref{S3a}. 
Further, using the Peres-Horodecki or PPT criterion~\cite{SepHHH,SepAP}, we show that irreducibly covariant channels for which the values of 
these parameters lie in the grey region shown in the right panel of \Cref{S3a}, are necessarily entanglement breaking (see \Cref{EB}).
\begin{figure}[h!]
	
	\begin{center}
		\subfloat[]{\includegraphics[width=0.48\textwidth]{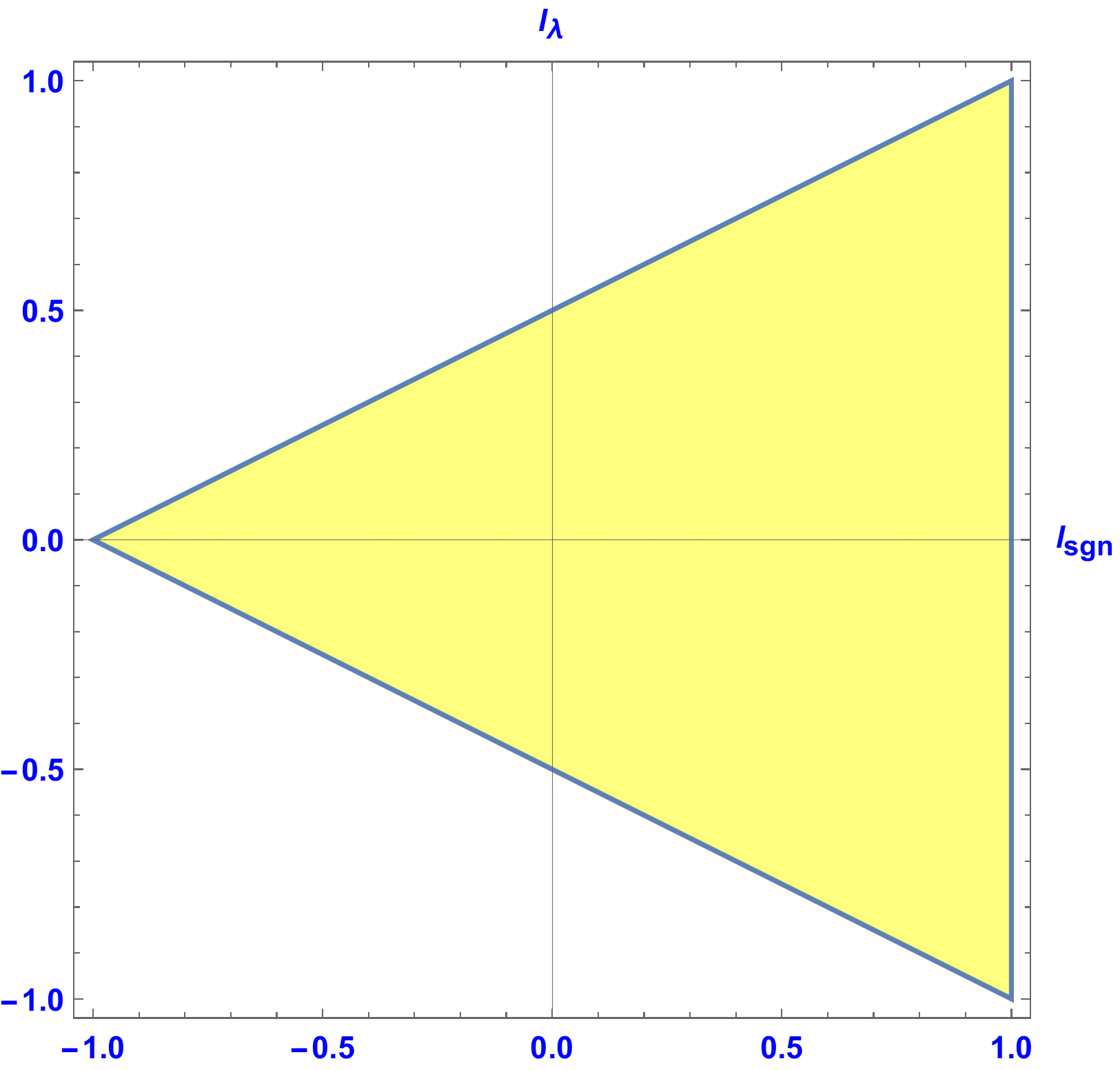}}
		\hfill
		\subfloat[]{\includegraphics[width=0.48\textwidth]{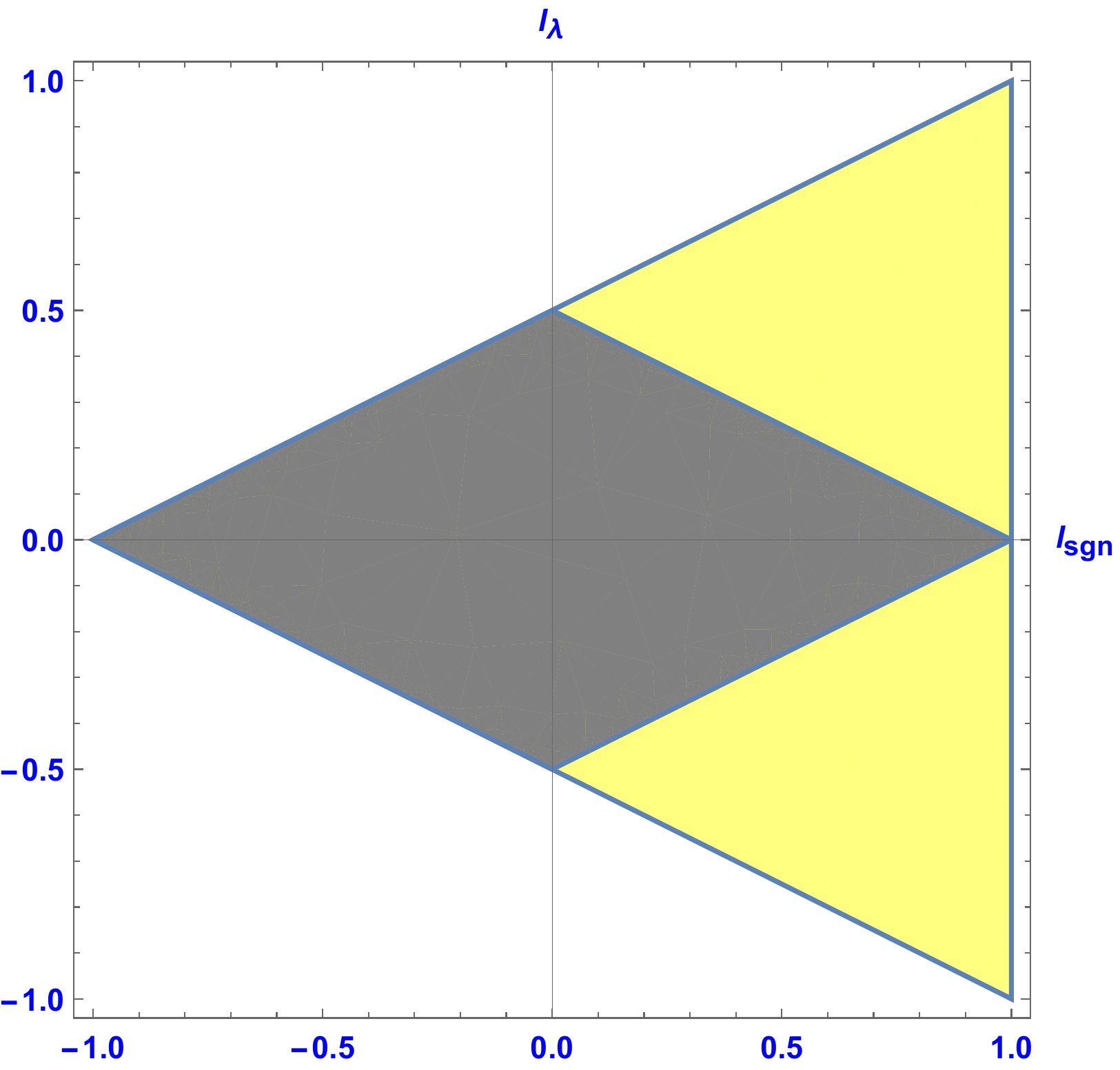}}
		\caption{The yellow triangle in the left panel denotes the 
			region of allowed values of the parameters $l_{\sgn}$ and $l_{\lambda}$, for which a trace-preserving linear map which is irreducibly covariant with respect to the $\epsilon-$representation of $S(3)$, for the partition $\lambda=(2,1)$, is a quantum channel (i.e.~is a CPTP map), and has Kraus operators given by \cref{expK}. The grey area on the right panel corresponds to the region in the parameter space for which the irreducibly covariant quantum channels are entanglement breaking.}
		\label{S3a}
	\end{center}
\end{figure}

\section{Notations and Definitions}
\label{nots}
We first define the concepts of vectorization and matrix representation of a linear map. 
Let $\M (n, \cC)$ denote the space of $n\times n$ complex matrices and let
$\{E_{ij}\}_{i,j=1}^n$, where $E_{ij} \equiv \ket{i}\bra{j}$, denote a basis
of $\M (n, \cC)$. Consider a matrix $A:= \left(a_{ij}\right)_{i.j=1}^n \equiv \left(a_{ij}\right)\in \M (n, \cC)$,
with $a_{ij}$ denoting the $(i,j)^{th}$ matrix element of $A$. The matrix $A$ can be expressed as a vector in $\cC^{n^2}$ using the standard matrix-vector isomorphism, referred to as {\em{vectorization}}.
\begin{definition}[Vectorization]
\label{prop2 copy(1)} Consider a map $\cV:\M(n,\mathbb{C})\rightarrow \mathbb{C}%
^{n^{2}}$, such that for any $A=\left( a_{ij}\right) \in \M(n,\mathbb{C})$ 
\begin{equation}
\cV(A)^{\tt}=\left(
a_{11},\ldots,a_{1n},a_{21},\ldots,a_{2n},\ldots,a_{n1},\ldots,a_{nn}\right)^T
\in \mathbb{C}^{n^{2}},
\end{equation}
where the superscript $^T$ denotes the transpose. In Dirac notation, $\cV(E_{ij}) \equiv \cV(\ket{i}\bra{j}) = \ket{i} \otimes \ket{j} \in \cC^{n^2}$. 
The map $\cV$  is an isomorphism between the linear spaces 
$\M(n,\mathbb{C})$ and $\mathbb{C}^{n^{2}}$, and we refer to it as {\em{vectorization}}.
\end{definition}

Consider a linear map $\Phi \in \End \left[ \M(n,\mathbb{C})\right]$, i.e.~$\Phi:\M(n, \mathbb{C})\rightarrow \M(n,\mathbb{C})$.
Its {\em{Choi-Jamio{\l}kowski image}} $J(\Phi)$ 
is given by~\cite{Jam,Choi}:
\begin{align}\label{choi}
J(\Phi)&:=\sum_{i,j=1}^nE_{ij}\otimes \Phi ( E_{ij}) \in \M(n^2 , \cC).
\end{align}
It is well-known that a linear map $\Phi \in \End \left[ \M(n,\mathbb{C})\right]$ is completely positive map if and only if its Choi-Jamio{\l}kowski image $J(\Phi)$ is a positive semidefinite matrix, i.e.~$J(\Phi)\geq 0$.

The matrix resulting from the action of $\Phi$ on any basis element $E_{ij} \in \M(n, \cC)$ can be expressed as follows:
\begin{equation}
\Phi (E_{ij})=\sum_{k,l=1}^n\phi _{kl,ij}E_{kl},
\end{equation}
The coefficients $\phi _{kl,ij}$ can be viewed as elements of an $n^2 \times n^2$ matrix, and hence we use
the notation:
\begin{align}
\mat(\Phi) &\equiv \left(\phi _{kl,ij}\right) \in \M(n^2, \cC).
\end{align}
Further, using the isomorphism $\M(n^{2},\mathbb{C})\simeq \M(n,\mathbb{C})\otimes
\M(n,\mathbb{C})$, the matrix $\mat(\Phi)=(\phi _{kl,ij})\in \M(n^{2},\mathbb{C}
) $ may be written in the form:
\begin{equation}\label{eq-mat}
\mat(\Phi )=\sum_{\nu=1}^{m}A^{\nu}\otimes \overline{B}^{\nu}; \quad \phi
_{kl,ij}=\left( \sum_{\nu=1}^{m}A^{\nu}\otimes \overline{B}^{\nu}\right) _{kl,ij}=
\sum_{\nu=1}^{m}a_{ki}^{\nu}\overline{b}_{lj}^{\nu},
\end{equation}
for some $A^{\nu}=(a_{ki}^{\nu})$, $B^{\nu}=(b_{lj}^{\nu})\in \M(n,\mathbb{C}),$ and $m \leq n^2$. In the above, we use the notation $\overline{
X}=(\overline{x}_{ij})$ to denote (element-wise) complex conjugation of a matrix $X\in \M(n,%
\mathbb{C})$, and $(X \otimes Y)_{kl,ij}=x_{ki}y_{lj}$ for any $X,Y \in
\M(n,\mathbb{C}).$ From this we easily obtain:
\begin{equation}
\Phi (E_{ij})=\sum_{\nu=1}^m\sum_{k,l=1}^na_{ik}^{\nu} E_{kl}\overline{b}_{lj}^{\nu},
\end{equation}
and 
\begin{align}
\label{phi-x}
\Phi
(X)=\sum_{\nu=1}^{m}A^{\nu}X\left( B^{\nu}\right)^{\dagger},
\end{align}
for any 
$
X=\sum_{ij}x_{ij}E_{ij}\in \M(n,\mathbb{C})$.

The above observations are summarized in the following lemma:
\begin{lemma}
	\label{L2}
For any $\Phi \in \End \left[ \M(n,\mathbb{C})\right]$, there exist sets of matrices $
A^{\nu}=( a_{ki}^{\nu}) ,B^{\nu}=( b_{lj}^{\nu}) \in \M(n,\mathbb{C}),$ with $\nu=1,\ldots,m$ (and $m \leq n^2$), such
that:
\begin{equation}
\forall \,X\in \M(n,\mathbb{C}),\qquad \Phi (X)=\sum_{\nu=1}^{m}A^{\nu}X\left( B^{\nu}\right)^{\dagger}.
\label{eq-8}
\end{equation}
\end{lemma}

\noindent
From \Cref{L2} and \cref{eq-mat} we get:
\begin{align}
\left(\Phi (X)\right)_{ij} =\sum_{\nu=1}^{m}\sum_{k,l=1}^n a_{ik}^{\nu}x_{kl} \overline{b}
_{lj}^{\nu}
&=\sum_{\nu=1}^{m}\sum_{k,l=1}^n a_{ik}^{\nu} \overline{b}
_{lj}^{\nu}x_{kl}\nonumber\\
&=\sum_{k,l=1}^n\sum_{\nu=1}^{m}\left( A^{\nu}\otimes \overline{B}
^{\nu}\right) _{il,kj}x_{kl}.
\end{align}
Further, using the vectorization map $\cV:\M(n,\mathbb{C})\rightarrow \mathbb{C}%
^{n^{2}}$, and \cref{eq-8}, we obtain:
\begin{equation}
\cV \left[ \Phi (X)\right] \equiv \cV \left( 
\sum_{\nu=1}^{m}A^{\nu}X \left( B^{\nu}\right)^{\dagger}\right)  =\sum_{\nu=1}^{m}\left( A^{\nu}\otimes \overline{B}^{\nu}\right) \cV (X),
\end{equation}
where $\sum_{\nu=1}^{m}\left( A^{\nu}\otimes \overline{B}^{\nu}\right) \in \M(n^{2},\mathbb{C})$
and $\cV (X)$, $\cV\left[ \Phi (X)\right] \in \mathbb{C}^{n^{2}}$. 
\medskip

\noindent
This leads to the following natural definition of the 
{\em{matrix representation}} of a linear map $\Phi \in \End \left[ \M(n,\mathbb{C})\right] $.

\begin{definition}[Matrix representation of a linear map]
	\label{matRep}
\smallskip
Let $\Phi \in \End \left[ \M(n,\mathbb{C})\right] $ such that $\forall X\in \M(n,\mathbb{C})$ $\
\Phi (X)=\sum_{\nu=1}^{m}A^{\nu}X\left( B^{\nu}\right)^{\dagger}$, for some matrices $
A^{\nu}=( a_{ki}^{\nu}) ,B^{\nu}=(b_{lj}^{\nu}) \in \M(n,\mathbb{C}),$ (for $\alpha=1,\ldots,m$). Then
the matrix $\mat(\Phi )\equiv \sum_{\nu=1}^{m}\left( A^{\nu}\otimes \overline{B}^{\nu}\right) \in
\M(n^{2},\mathbb{C})$ is a matrix representant of the map $\Phi$. For this
representation we have:
\begin{equation}
\left[ \mat(\Phi )\cV(X)\right]_{ij}=\Phi (X)_{ij},\qquad i,j=1,\ldots,n.
\end{equation}
	\label{def4n}
\end{definition}
The vector space $\M(n,\mathbb{C})$ becomes a Hilbert space when equipped with the Hilbert-Schmidt scalar product: 
\be
\forall X,Y\in \M(n,\mathbb{C})\qquad (X,Y)\equiv \tr(X^{\dagger}Y)
\ee
with induced Hilbert-Schmidt norm:
\be
 ||X||_2\equiv \sqrt{\tr\left( X^{\dagger}X \right)}.
\ee
A linear map $\Phi \in \End[\M(n,\mathbb{C})]$ acts as an operator on this Hilbert space, and its adjoint $\Phi^{\ast}$ is defined through the relation:
\be
\forall X,Y\in \M(n,\mathbb{C})\qquad (\Phi ^{\ast }(X),Y)=(X,\Phi (Y)).
\ee
Hence, from \cref{phi-x} and the above, it follows that:
	\be\label{adjt}
	\Phi ^{\ast }(X)=\sum_{\nu =1}^{m}(A^{\nu })^{\dagger}XB^{\nu },\qquad X\in \M(n,\mathbb{C}).
	\ee
Further, the following proposition holds.
\begin{lemma}
	\label{Pa}
	Taking the adjoint of an element in $\End[\M(n,\mathbb{C})]$ induces Hermitian conjugation of its matrix representant in the space 
	$\M(n^{2},\mathbb{C})$. Namely, for any $\Phi \in \End[\M(n,\mathbb{C})]$ we have
	\be
	\mat(\Phi ^{\ast })=(\mat(\Phi ))^{\dagger}.
	\ee
\end{lemma}

\begin{proof}
	From \cref{adjt} we have
	\be
	\Phi ^{\ast }(X)=\sum_{\nu =1}^{m}(A^{\nu })^{\dagger}XB^{\nu }=\sum_{\nu
		=1}^{m}(A^{\nu })^{\dagger}X((B^{\nu })^{\dagger})^{\dagger}. 
	\ee
Then, using~\Cref{matRep} we get 
	\be
	\mat(\Phi ^{\ast })=\sum_{\nu =1}^{m}(A^{\nu })^{\dagger}\otimes \overline{((B^{\nu
		})^{\dagger})}=\sum_{\nu =1}^{m}(A^{\nu }\otimes \overline{(B^{\nu })}
	)^{\dagger}=(\mat(\Phi ))^{\dagger}. 
	\ee
\end{proof}

\section{Irreducibly covariant linear maps and quantum channels}
\label{irrCov}
Let $G$ be a finite group and let:
\begin{equation}
U:G\rightarrow \M(n,\mathbb{C}),\qquad{\hbox{i.e.~}}\,\, U(g)=\left( u_{ij}(g)\right) \in \M(n,\mathbb{C})\,\, \  \forall g \in G,
\end{equation}
be a unitary irreducible representation (irrep, in short) of $G$. The contragradient representation
$U^c: G \rightarrow \M(n,\mathbb{C})$ is given by
\begin{align}\label{contra}
U^c(g) &= U(g^{-1})^T \equiv {\overline{U}}(g) \quad \forall \, g \in G.
\end{align}

The map $
\Ad_{U}^{G}:G\longrightarrow \End\left[ \M(n,\mathbb{C})\right] $ 
is called the {\em{adjoint representation}} of the group $G$
with respect to the unitary irrep $U$, and is defined through its action 
on any $X \in \M(n , \cC)$ as follows:
\begin{align}
\Ad_{U(g)}(X)& \equiv U(g)XU^{\dagger }(g) \quad \,\, \forall g \in G.
\end{align}
Obviously, $ \Ad
_{U(g)}\in \End \left[ \M(n,\mathbb{C})\right].$

\begin{definition}[Commutant of the adjoint representation] 
\label{def-intertwine}
Let $\Int_G(\Ad_U)$ denote the set of {\em{intertwiners}} of $\Ad_U$, i.e.~the set of maps in $\End\left[ \M(n,\mathbb{C})\right]$ whose
action commutes with that of  $\Ad_U$:
\begin{align}
\Int_G(\Ad_U)&
=\{\Psi \in \End\left[ \M(n,\mathbb{C})\right] : \Psi \circ \Ad_U=\Ad_U \circ \Psi\}.
\end{align}
\end{definition} 
\smallskip

\noindent
Note that for any  $\Ad_{U(g)}\in $ $\End[\M(n,\mathbb{C})]$ (with $g\in G$) $\forall X=\left( x_{ij}\right) \in \M(n,\mathbb{C})$, we have:
\begin{equation}
 \label{eq5}
\left( \Ad_{U(g)}(X)\right)
_{ij}=\left( U(g)XU^{\dagger}(g)\right)_{ij}=\sum_{kl}\left[ U(g)\otimes \overline{U}%
(g)\right]_{il,kj}x_{kl}, 
\end{equation}%
so that:
\begin{align}
\mat(\Ad_{U(g)})=U(g)\otimes \overline{U}(g).
\end{align}%
Thus the operator $\Ad_{U(g)}\in \End[\M(n,\mathbb{C})]$ may be
represented as a matrix $U(g)\otimes \overline{U}(g)\in \M(n^{2},\mathbb{C})$.
Using the definition of the character of a representation~\cite{NS}, we arrive at the following lemma.
\begin{lemma}
	\label{Prop5n}
	The character of the adjoint representation $\Ad_{U}^{G}:G
	\longrightarrow \End\left[ \M(n,\mathbb{C})\right] $, denoted 
by $\chi^{\Ad_{U}^{G}} \equiv \chi ^{\Ad}:G\longrightarrow 
	\mathbb{C}$, is given by
	\begin{align}
	\chi^{\Ad_{U}^{G}}(g) &:= \tr \left(\mat(\Ad_{U(g)})\right)= \tr \left( U(g)\otimes \overline{U}(g)\right), \quad \forall g \in G.
\end{align}
In particular, we have
\begin{align}\label{char-ad}
\chi ^{\Ad}(g)&=|\chi ^{U}(g)|^{2}, \quad \forall g \in G,
	\end{align}
where $\chi ^{U}:G\rightarrow \mathbb{C}$ is the character of the
representation $U:G\rightarrow \M(n,\mathbb{C})$, i.e.~$\chi ^{U}(g)=\tr\left(
U(g)\right) $, $\forall g \in G$.
\end{lemma}

\begin{definition}[Irreducibly covariant- linear maps (ICLM) and quantum channels (ICQC)]
\label{def1}
A linear map $\Phi \in \End \left[ \M(n,\mathbb{C})\right]$ is said to be irreducibly covariant
with respect to the unitary irrep $U:G\rightarrow \M(n,\mathbb{C})$ of a 
finite group $G$, if
\begin{equation}
\forall g\in G,\quad \forall X\in \M(n,\mathbb{C})\quad \Ad_{U(g)}[\Phi
(X)]=\Phi \lbrack \Ad_{U(g)}(X)],
\end{equation}
i.e.~$\Phi \in \Int_{G}(\Ad_{U})$.
Further, if the linear map $\Phi$ is completely positive and trace-preserving, then it is referred
to as an irreducibly covariant quantum channel. We denote an irreducibly covariant linear map
by the acronym ICLM, and an irreducibly covariant quantum channel by the acronym ICQC.
\end{definition}

\begin{remark}
	\label{R-ICLM}
	Note that the set of $ICLM$s form an algebra with composition of linear maps as a product.
\end{remark}

\section{Tools and results from group representation theory}
\label{S4}
Using tools from group representation theory, it is easy to prove the following proposition:
\begin{lemma}
\label{def7 copy(1)} A linear map $\Phi \in \End\left[ \M\left( n,\mathbb{C}\right)\right]$ is irreducibly covariant with respect to the irrep $
U:G\rightarrow \M(n,\mathbb{C})$ of a finite group $G$ (i.e.~$\Phi \in \Int_{G}(\Ad_{U}))$ if and only if 
\begin{align}
\mat(\Phi )& \in \Int_{G}\left( U\otimes
U^{c}\right) =\Int_{G}\left( U\otimes \overline{U}\right).
\end{align}
\end{lemma}
For sake of completeness, we include the proof of this lemma in~\Cref{appA}. The above lemma implies that
instead of studying the structure of the commutant $
\Int_{G}(\Ad_{U})$ in the space $\End\left[\M\left( n,\mathbb{C}\right)\right]$, it suffices to
study the structure of the commutant $\Int_{G}\left( U\otimes
U^{c}\right) $ in the matrix space $\M(n^{2},\mathbb{C})$, which is simpler
because we have to deal with matrices only.

It is known that the representation~\cite{NS}:
\begin{equation}
U\otimes U^{c}:G\rightarrow \M(n^{2},\mathbb{C})
\end{equation}%
is not irreducible and we have:
\begin{equation}
\label{gen_decomp}
U\otimes U^{c}=\bigoplus _{\alpha }m_{\alpha }\varphi ^{\alpha },
\end{equation}%
where $\varphi ^{\alpha}$ are unitary irreps of the group $G$: $\forall \, g \in G$, $%
\varphi ^{\alpha}(g)=\left( \varphi _{ij}^{\alpha }(g)\right) \in \M\left(
\left| \varphi^{\alpha }\right| ,\mathbb{C}\right) $,  and $m_{\alpha }$  is the multiplicity of the irrep $\varphi^{\alpha }$ of dimension $\left|\varphi^{\alpha } \right|\equiv \dim \varphi^{\alpha }$. The multiplicity $m_{\alpha}$ is given by the 
following expression~\cite{NS}:
\begin{equation}
\label{mult}
m_{\alpha }=\frac{1}{|G|}\sum_{g\in G}\chi ^{\alpha }\left( g^{-1}\right) \chi ^{\Ad}(g),
\end{equation}
where $\chi ^{\alpha }$ denotes the character of the irrep $\varphi ^{\alpha }(g) $,  and $\chi^{\Ad}(g)=\left|\chi^U(g) \right|^2$ is the   character of the
adjoint representation $\Ad_{U}^{G}$ defined in~\Cref{Prop5n}.  From this it follows that the commutant
of the representation $U\otimes U^{c}$:
\begin{equation}
\Int_{G}\left( U\otimes U^{c}\right) =\left\{ A\in \M\left(
n^{2},\mathbb{C}\right) :\forall g\in G\quad A\left(U(g)\otimes 
\overline{U}(g)\right)=\left(U(g)\otimes \overline{U}(g)\right)A\right\} 
\end{equation}%
is nontrivial i.e.~it is not one-dimensional. 
\smallskip

In fact, from the theory of group representations~\cite{NS}, one can deduce the following:
\begin{lemma}
\label{prop3} Let $U:G\rightarrow \M(n,\mathbb{C})$ be a
unitary irreducible representation of a given finite group $G$. Then we have
\begin{equation}
U\otimes U^{c}=\varphi ^{\id}\oplus _{\alpha \neq \id}m_{\alpha
}\varphi ^{\alpha },
\end{equation}%
i.e.~the identity irrep, $\varphi ^{\id}$, is always included
in the representation $U\otimes U^{c}$ with multiplicity one. Moreover, 
\begin{equation}
\dim \left[ \Int_{G}\left( U\otimes U^{c}\right) \right] =\frac{1}{|G|}%
\sum_{g\in G}\left\vert \chi ^{U}(g)\right\vert ^{4},
\end{equation}%
where $\chi ^{U}:G\rightarrow \mathbb{C}$ is the character of the
representation $U:G\rightarrow \M(n,\mathbb{C})$, and $|G|$ is the cardinality of the group $G$.
\end{lemma}
\bigskip

\noindent
We illustrate the above lemma by the following two examples for the symmetric group $S(n)$ for $n=3, 4$. For the basics and notations of the representation theory of $S(n)$ see for example~\cite{JJF}.
\begin{example}
\label{ex5} For $G=S(3)$ and its two-dimensional, unitary irrep $U=\varphi ^{(2,1)}$
characterised by the partition $\lambda =(2,1)$ we have
\begin{equation}
U\otimes U^{c}=\varphi ^{\id}\oplus \varphi ^{\sgn}\oplus
\varphi ^{(2,1)},\quad \dim \left[ \Int_{S(3)}\left( U\otimes
U^{c}\right) \right] =3.
\end{equation}
\end{example}

\begin{example}
\label{ex6} For $G=S(4)$ and its two-dimensional, unitary irrep $U=\varphi ^{(2,2)}$
characterised by the partition $\lambda =(2,2)$ we have
\begin{equation}
U\otimes U^{c}=\varphi ^{\id}\oplus \varphi ^{\sgn}\oplus
\varphi ^{(2,2)},\quad \dim \left[ \Int_{S(4)}\left( U\otimes
U^{c}\right) \right] =3.
\end{equation}
\end{example}

In these examples all the multiplicities are equal to one so that the
decomposition of the representation $U\otimes U^{c}$ is simply reducible. This holds
also for all other irreps of the group $G=S(4)$.
However, in the case
of the group $S(5)$ there are irreps  $\varphi^{\alpha}$ for which 
the representation $U\otimes U^c$ is not simply reducible. 
\medskip

\noindent
{\bf{Notation}}
\smallskip

For notational simpliciity, henceforth, elements of $\End\left[ \M(n,\mathbb{C})\right]$, will be denoted by Greek letters,
whereas their matrix representations in $\M(n^2, \cC)$ will be denoted by the same Greek letters but with a tilde.  For example,
\begin{align}
\Gamma &=\sum_{g\in G}a_{g}\Ad_{U(g)}\in \End\left[ \M(n,\mathbb{C})\right], \quad a_g\in \mathbb{C},
\end{align}%
and 
\begin{align}
\tGamma\equiv \mat(\Gamma) & =\sum_{g\in G}a_{g}U(g)\otimes \overline{U}(g)\in 
\M(n^2, \mathbb{C}). 
\end{align}
The Schur orthogonality relations~\cite{FHa}, given below, are useful in proving the results in the following sections:
\begin{itemize}
\item{\begin{align}
	\label{or}
	\sum_{g\in G}\varphi _{ij}^{\alpha }\left( g^{-1}\right) \varphi _{kl}^{\beta }(g)&=\frac{
		|G|}{|\varphi^{\alpha}|}\delta^{\alpha \beta}\delta _{jk}\delta _{il},
	\end{align} where $|G|$ denotes cardinality of the group $G$.}
\item{\begin{align}
	\label{orr2}
	\frac{1}{|G|}\sum_{g\in G}\chi ^{\alpha }\left( g^{-1}\right)&=
	\begin{cases}
		1 \ if \ \alpha =\id,\\ 
		0 \ if~\ \alpha \neq \id,
	\end{cases}
	\end{align} where $\id$ denotes the identity irrep of the group $G$.}
\item
\begin{align}
\label{or2}
\sum_{\alpha \in \widehat{G}}\chi^{\alpha }(h)\chi^{\alpha }\bigl( g^{-1}\bigr) 
&=\frac{|G|}{|K(h)|}\delta_{K(h)K(g)},
\end{align}
where by $K(g)$ we denote the conjugacy class of the element $g \in G$,
\begin{align}
\label{conju}
K(g)=\left\lbrace a\in G \ | \ \exists h\in G \ \text{with} \ a
=hgh^{-1}\right\rbrace,
\end{align}
and sum in the~\cref{or2} runs over the set of all inequivalent irreps $\widehat{G}$.
\end{itemize}
\section{Intermediate Results}
\label{S5}
In this section we state some propositions and corollaries which are required to prove our main results. Their proofs are given in~\Cref{appB}. 
\begin{proposition}
	\label{prop8} Suppose that an unitary irrep $U:G\rightarrow \M(n,\mathbb{C})$, of a finite
	group $G$ is
	such that $U\otimes U^{c}$ is multiplicity-free, i.e.
	\begin{equation}
	U\otimes U^{c}=\bigoplus_{\alpha \in \Theta }\varphi ^{\alpha },
	\label{eq15}
	\end{equation}%
	where $\Theta $ is the index set of those irreps ($\varphi^\alpha$) of the group $G$ which
	appear (with multiplicity one) in the above decomposition of $U\otimes U^{c}$. Then,
	\begin{equation}
	\label{span_setM}
	\Int_{G}\left( U\otimes U^{c}\right) =\operatorname{span}_{\mathbb{C}}\left\{ 
	{\tPi}^{\alpha}:\alpha \in \Theta \right\} \quad {\hbox{and}} \quad\dim
	\left[ \Int_{G}\left( U\otimes U^{c}\right)\right]  =|\Theta |,
	\end{equation}
	where 
	\begin{equation}
	{\tPi}^{\alpha }=\frac{\left|\varphi ^{\alpha }\right|}{|G|}\sum_{g\in G}\chi
	^{\alpha }\left( g^{-1}\right) U(g)\otimes \overline{U}(g)\in \M( n^{2},%
	\mathbb{C}) .  \label{eq17}
	\end{equation}%
	The matrices ${\tPi}^{\alpha }$ have the following properties: 
	\begin{equation}
	{\tPi}^{\alpha }{\tPi}^{\beta }=\delta_{\alpha \beta }\tPi%
	^{\alpha },\qquad (\tPi^{\alpha })^{\dagger }=\tPi^{\alpha
	},\qquad \sum_{\alpha \in \Theta }\tPi^{\alpha }=\text{\noindent
	\(\mathds{1}\)}_{n^{2}},
	\label{eq18}
	\end{equation}%
	where $\text{\noindent
		\(\mathds{1}\)}_{n^{2}}$ is the identity operator on $\cC^{n^2}$. The set
	$\left\{{\tPi}^{\alpha }\,:\, \alpha \in \Theta\right\}$ is a complete set of orthogonal projectors and $\tr\tPi^{\alpha
	} =\left| \varphi^{\alpha }\right|$. 
\end{proposition}

\begin{corollary}
\label{Rprop8}
The identity irrep $\id$ always occurs in the decomposition~\cref{eq15}, so $\id \in \Theta$. Moreover 
\be
\label{++}
\gamma \notin \Theta \Rightarrow \widetilde{\Pi}^{\gamma}=0,
\ee
where $\widetilde{\Pi}^{\gamma}$ are the projectors defined through~\cref{eq17}.
\end{corollary}	
This can be easily seen as follows: One can notice, that from~\cref{eq17} we can deduce that:
\be
\tr (\widetilde{\Pi}^{\gamma})=m_{\gamma}|\varphi^{\gamma}|,
\ee
where $m_{\gamma}$ is the multiplicity given in~\cref{mult}, therefore if $\gamma \notin \Theta$, then $m_{\gamma}=0$ and $\tr (\tPi^{\gamma})=0$. The latter implies, that $\tPi^{\gamma}=0$ because $\tPi^{\gamma}$ is a projector.

From~\Cref{prop8} we see why an assumption of the multiplicity freeness is so useful. This is because in the general scenario, i.e.~when multiplicities occurring in decomposition~\eqref{gen_decomp} satisfy $m_{\alpha} \geq 1$,   the spanning set of $\Int_{G}\left( U\otimes U^{c}\right)$ is larger than the set given through expression~\cref{span_setM}, since the multiplicity space is not trivial any more.
From the \Cref{prop8} and \Cref{def7 copy(1)} we get the following corollary.
\begin{corollary}
	\label{cor10} A linear map $\Phi \in \End\left[ \M( n,\mathbb{C})\right] $, which is irreducibly covariant with respect to a unitary irrep $U:G\rightarrow
	\M(n,\mathbb{C})$ of a finite group $G$, can be expressed in the form 
	\begin{equation}
	\Phi =l_{\id}\Pi^{\id}+\sum_{\alpha \in \Theta ,\alpha \neq \id}l_{\alpha }\Pi^{\alpha }:\quad l_{\alpha }\in \mathbb{C},  \label{eq21}
	\end{equation}
	where 
	\begin{equation}
	\label{Pi_a}
	\Pi^{\alpha }=\frac{\left| \varphi^{\alpha }\right|}{|G|}\sum_{g\in G}\chi ^{\alpha
	}\left( g^{-1}\right) \Ad_{U(g)}\in \End\left[ \M\left( n,\mathbb{C}\right)\right] ,\quad
	\alpha \in \Theta, 
	\end{equation}
	and the operators $\Pi^{\alpha }$ have the same properties as their matrix
	representants $\tPi^{\alpha } \equiv \mat (\Pi^\alpha)$, i.e.
	\begin{equation}
	\label{properties1}
	\Pi^{\alpha }\Pi^{\beta }=\delta_{\alpha \beta }\Pi^{\alpha },\qquad (\Pi^{\alpha
	})^{\ast }=\Pi^{\alpha },\qquad \sum_{\alpha \in \Theta }\Pi^{\alpha
}=\id_{\End\left[ \M(n,\mathbb{C})\right] },
\end{equation}
where $\id_{\End\left[ \M(n,\mathbb{C})\right] }$ denotes the identity map in ${\End\left[ \M(n,\mathbb{C})\right] }$.
\end{corollary}

\begin{remark}
The expression given in~\cref{eq21} of~\Cref{cor10} is the spectral decomposition of $\Phi $: the coefficients $l_{\alpha }$ are
its eigenvalues, with $\Pi^\alpha$ being the corresponding projectors. 
\end{remark}

In the next step we will need the following statement describing the structure of the projectors $\widetilde{\Pi}^{\alpha}$ which span the commutant $\Int_{G}\left( U\otimes U^{c}\right) $.
\begin{proposition}[Spectral decomposition of the projectors $\tPi^{\alpha}$]
	\label{prop11} Let $\tPi^{\alpha }$ be a projector as in Proposition~\ref{prop8}.
	It has the following spectral decomposition:
	\begin{equation}
	\label{propertiesE}
	\tPi^{\alpha}=\sum_{i=1}^{\left| \varphi^{\alpha }\right|}\tPi_{i}^{\alpha },\qquad
	\tPi_{i}^{\alpha }\tPi_{j}^{\beta }=\delta^{\alpha \beta}\delta _{ij}\tPi_{i}^{\alpha },\quad \left(
	\tPi_{i}^{\alpha }\right) ^{\dagger }=\tPi_{i}^{\alpha },\quad \tr\left(
	\tPi_{i}^{\alpha }\right) =1,
	\end{equation}%
	where 
	\begin{equation}
	\label{expF}
	\tPi_{i}^{\alpha }=\frac{\left| \varphi^{\alpha }\right|}{|G|}\sum_{g\in G}\varphi
	_{ii}^{\alpha }\left( g^{-1}\right) U(g)\otimes \overline{U}(g)\in \M(n^{2},%
	\mathbb{C}).
	\end{equation}
\end{proposition}

\begin{remark}
	\label{Rprop11}
	In particular as a remark of~\Cref{prop11} we can say that for $p,q,s,t\in\{1,\ldots,n\}$ we have:
	\be
	\widetilde{\Pi}^{\id}=\left( \widetilde{\Pi}^{\id}\right) _{pq,st}=\left( \frac{1}{|U|}\delta _{pq}\delta
	_{st}\right) \Rightarrow \left( \widetilde{\Pi}^{\id}\right) _{pq,pq}=\left( \frac{1}{|U|}\delta
	_{pq}\right) ,
	\ee
	where $|U|\equiv \dim U$. 
\end{remark}

Hence, the spectral decomposition of $\tPi^{\alpha }$ is
explicitly given by characteristics of the representation $%
U\otimes U^{c}$ of the group $G$ (which are known for any given irrep $U$).

\begin{corollary}
	\label{cor17}
	The linear maps $\Pi^{\alpha}_i\in \End\left[ \M(n,\mathbb{C})\right]$ such that $\mat\left[\Pi^{\alpha}_i \right] ={\tPi}_i^{\alpha}$, are given by:
	\be
	\label{Pii_a}
	\Pi^{\alpha}_i=\frac{\left| \varphi^{\alpha }\right|}{|G|}\sum_{g\in G}\varphi
	_{ii}^{\alpha }\left( g^{-1}\right) \Ad_{U(g)}.
	\ee 
	and satisfy relations analogous to~\cref{propertiesE} of their matrix representants ${\tPi}_i^{\alpha}$:
	\begin{equation}
	\label{chain}
	\Pi^{\alpha}=\sum_{i=1}^{\left| \varphi^{\alpha }\right|}\Pi_{i}^{\alpha },\qquad
	\Pi_{i}^{\alpha }\Pi_{j}^{\beta }=\delta^{\alpha \beta}\delta _{ij}\Pi_{i}^{\alpha },\quad \left(
	\Pi_{i}^{\alpha }\right) ^{\ast}=\Pi_{i}^{\alpha }.
	\end{equation}
\end{corollary}

The following proposition gives the spectral decomposition of the rank one
projectors $\Pi_{i}^{\alpha } \in \End\left[ \M(n, \cC)\right]$.

\begin{proposition}
	\label{propP} Let $V_{i}^{\alpha }\in \M(n,\mathbb{C})$ denote the normalised in the Hilbert-Schmidt norm
	eigenvectors of the projector $%
	\Pi_{i}^{\alpha }\in \End\left[ \M\left( n,\mathbb{C}\right)\right] $, corresponding to the eigenvalue $1$, i.e.
	\begin{align}
	\Pi_i^\alpha V_{j}^{\beta } &= \delta^{\alpha \beta}\delta_{ij}V_{i}^{\alpha }.
	\end{align}
	Then $V_{i}^{\alpha }$ has the following form: there exists a pair $(s,t)$ with $s,t \in \left\lbrace 1,\ldots ,n\right\rbrace$ such that:
	\begin{equation}
	\label{11}
	\quad V_{i}^{\alpha }\equiv V_{i}^{\alpha }(s,t)= \frac{1}{\sqrt{\left(\tPi_i^\alpha\right)_{st,st}}}\frac{\left| \varphi^{\alpha }\right|}{|G|}\sum_{g\in G}\varphi _{ii}^{\alpha
	}\left( g^{-1}\right) U_{C(s)}(g)U_{R(t)}\left( g^{-1}\right)
	\neq 0,
	\end{equation}
	where $U_{C(s)}(g)$ and $U_{R(t)}(g)$ respectively denote the $s^{th}$ column and the $t^{th}$ row of the matrix $U(g) \in \M(n, \cC)$. 
	If  $(s,t)$ and $(p,q)$ (with $s,t,p,q\in \{1,\ldots,n\}$) are pairs for which $V^{\alpha}_i(s,t)\neq 0$ and $V^{\alpha}_i(p,q)\neq 0$, then the following orthonormality relation holds:
	\be
	\label{VHS}
	(V_{i}^{\alpha }(s,t),V_{j}^{\beta }(p,q))=\delta ^{\alpha \beta }\delta
	_{ij}e^{\operatorname{i}\zeta },
	\ee
	where $\operatorname{i}^2=-1$ and $\zeta$ is some phase factor. 
\end{proposition}

\begin{corollary}
	\label{RpropP}
	In~\cref{VHS} of~\Cref{propP}, the phase factor $\zeta $ depends on the indices $s,t,p,q\in
	\{1,...,n\}$ and if $(s,t)=(p,q)$ then $\zeta =0$, so that
	\be
	\label{astt}
	V_{i}^{\alpha }(p,q)=e^{\operatorname{i}\zeta }V_{i}^{\alpha }(s,t).
	\ee
\end{corollary}
	Note that from~\cref{11} of~\Cref{propP} we can deduce that
	\be
	\label{RpropP2}
	V^{\id}(s,s)=\frac{1}{\sqrt{|U|}}\text{\noindent
		\(\mathds{1}\)}_{n}.
	\ee

For any $\beta \in \Theta$ and $i\in \{1,2,\ldots,|\varphi^{\beta}|\}$, let us define the set
\begin{align}\label{setS}
\cS_{\beta,i} := \left\{ (s,t)\in \{1,\ldots,n\}\times \{1,\ldots,n\}\,:\, 
\left(\widetilde{\Pi}_i^{\beta} \right)_{st,st}\neq 0 \right\}.
\end{align}
Then for any $\beta \in \Theta$ and $i\in \{1,2,\ldots, \left| \varphi^\beta\right|\}$ the vector $V_{i}^{\beta}$ is uniquely parametrized (up to a phase) by a given pair in the set $\cS_{\beta,i}$. The phase turns out to be irrelevant in our characterization of irreducibly covariant linear maps or quantum channels (see \Cref{independence}). The set $\cS_{\beta,i}$ parametrizes the non-zero vectors $V_{i}^{\beta }(s,t)$ because we
have (see~\cref{thm-30-1} in~\Cref{thm21}):
\be
\left| \left| V_{i}^{\beta }(s,t)\right| \right| _{2}^{2}=\left( \widetilde{\Pi} _{i}^{\beta }\right) _{st,st},
\ee
so if $\left( \widetilde{\Pi} _{i}^{\beta }\right) _{st,st}\neq 0$ then the coreponding vector $V_{i}^{\beta }(s,t)$ is well defined.
\begin{corollary}
	\label{n2}
	The set of $n^{2}$ matrices
\be
\label{sV}
\left\{V_{i}^{\beta} \equiv V_i^\beta(s,t)\,:\,(s,t) \in \cS_{\beta,i}, \   \beta \in \Theta, \  i\in \{1,2,\ldots, | \varphi^\beta|\}\right\},
\ee
constitute an orthonormal basis of the
	linear space $\M(n,\mathbb{C})$. 
\end{corollary}

\begin{remark}
	In order to construct the basis given in~\cref{sV} of~\Cref{n2} one has to construct the projectors $\widetilde{\Pi} _{i}^{\beta }$ defined in~\cref{expF} of~\Cref{prop11}, and then choose the indices $(s,t)$ such that $\left( \widetilde{\Pi}_{i}^{\beta }\right) _{st,st}\neq 0$.
	 In particular $\forall \beta \in \Theta$ and $\forall i=1,\ldots,|\varphi^{\beta}|$ by fixing a pair $(s,t)$ from the above set we also fix the basis.
\end{remark}
From \Cref{propP},~\Cref{prop2 copy(1)} and~\Cref{matRep} we obtain the following corollary:
\begin{corollary}
	\label{corr16}
	The vector $\cV(V_{i}^{\alpha })\in\mathbb{C}^{n^{2}}$ is an eigenvector of $\mat({\Pi}_{i}^{\alpha })\equiv \widetilde{\Pi}_{i}^{\alpha }\in \M\left( n^{2},\mathbb{C}\right) $
with eigenvalue $1$, i.e.
	\begin{equation}
	\widetilde{\Pi}_{i}^{\alpha }\cV(V_{i}^{\alpha })=\cV(V_{i}^{\alpha }).
	\end{equation}
\end{corollary}

\medskip

\noindent
The following proposition gives a necessary and sufficient condition for 
an irreducibly covariant linear map (ICLM) to be trace-preserving.
\begin{proposition}
	\label{prop12} An ICLM $\Phi =l_{\id}\Pi^{\id}+\sum_{\alpha \in
		\Theta ,\alpha \neq \id}l_{\alpha }\Pi^{\alpha }\in \Int_{G}\left(
	\Ad_{U}\right) $ is trace preseving if and only if $l_{\id}=1$, so that it is
	of the form:
	\begin{equation}\label{decomp}
	\Phi =\Pi^{\id}+\sum_{\alpha \in \Theta ,\alpha \neq \id}l_{\alpha
	}\Pi^{\alpha },
	\end{equation}%
	where the coefficient $l_{\alpha }$ for $\alpha \in \Theta$, with $\alpha \neq \id$, can be arbitrary.
\end{proposition}

To establish a necessary and sufficient condition for an irreducibly covariant linear (ICLM) map 
$\Phi \in \End \left[ \M(n, \cC)\right] $ to be completely positive, it is convenient to consider the Choi-Jamio{\l}kowski image $J(\Phi)$ 
of the map defined through \cref{choi}, since the complete positivity of $\Phi$ is equivalent to positive semi-definiteness of 
$J(\Phi)$. Restricting ourselves an ICLM, $\Phi$, which is trace-preserving, by the linearity of the Choi-Jamio{\l}kowski transformation we get
\begin{equation}\label{CJ-decomp}
J(\Phi)=J\left( \Pi^{\id}\right) +\sum_{\alpha \in \Theta ,\alpha \neq 
	\id}l_{\alpha }J(\Pi^{\alpha }),
\end{equation}
an explicit form for which is given in the following proposition.
\begin{proposition}
	\label{prop13} 
The Choi-Jamio{\l}kowski image of a trace-preserving ICLM $\Phi \in \End\left[ \M(n, \cC)\right]$ (as given by \Cref{prop12})
is given by 
	\begin{equation}
	\label{71}
	J(\Phi )=\frac{1}{|U|}\text{\noindent
		\(\mathds{1}\)}_{n}\otimes \text{\noindent
		\(\mathds{1}\)}_{n} +\frac{1}{|G|}%
	\sum_{ij}E_{ij}\otimes \sum_{g\in G}\left( \sum_{\alpha \in \Theta ,\alpha
		\neq \id}l_{\alpha }\left|\varphi^{\alpha} \right| \chi ^{\alpha }\left(
	g^{-1}\right) \right) U_{C(i)}(g)\left( U_{C(j)}(g)\right) ^{\dagger },
	\end{equation}%
	where $U_{C(i)}(g)=\left( u_{ki}(g)\right)_k$
	denotes the $i^{\text{th}}$ column of the matrix $U(g)$.
\end{proposition}

	\begin{remark}
		\label{Ch-slad}
	The trace of $J(\Phi )$ from~\cref{71} of~\Cref{prop13} depends only on the dimension of the irrep $U$,
	and is independent of the ICLM $\Phi$, i.e.
	\begin{equation}
	\label{t-Choi}
	\tr\left(J( \Phi) \right)=|U|.
	\end{equation}%
	\end{remark}
The above relation can be easily  obtained by a direct calculation of the trace, using the orthogonality relation, \cref{orr2} for irreducible characters.

The following proposition deals with the eigenvalue problem of the Choi-Jamio{\l}kowski images
of the projectors $\Pi^\alpha$ arising in the decomposition in~\cref{decomp} of the trace-preserving ICLM $\Phi$.
It provides the first step towards finding a necessary and sufficient condition for the positive semi-definiteness of $J(\Phi)$. 
Let us start from the following definition
\begin{definition}
	\label{dprop14}
Let $V_{i}^{\alpha } \in \M(n , \cC)$ be the
normalised eigenvectors of the operators $\Pi_{i}^{\alpha }\in \End\left[ \M( n,
\mathbb{C})\right] $ (given in Proposition \ref{propP}), 
which form an orthonormal basis of $\M(n,\mathbb{C})$ (as stated in \Cref{n2}). Let us define the set of  $n^2$ vectors
\begin{equation}
\label{malutkie}
|v_{i}^{\beta }\>  \equiv \sum_{k,l=1}^{n}\left( V_{i}^{\beta }\right)
_{kl}\cV(E_{lk})\in \mathbb{C}^{n^{2}},\quad \beta \in \Theta ,\quad i=1,\ldots ,|\varphi^{\beta}|.
\end{equation}
\end{definition}
Using arguments similar to those used in the proof of~\cref{VHS} of~\Cref{propP}, we can show that the vectors defined by~\cref{malutkie} of~\Cref{dprop14}
satisfy the following orthonormality relation:
\begin{lemma}
	\label{vec_v}
	For a given basis $\{V_{i}^{\alpha }\equiv V_{i}^{\alpha
	}(s_{\alpha },t_{\alpha }):\alpha \in \Theta ,$ $i=1,\ldots,|\varphi^{\alpha} |\}$ of the
	linear space $\M(n,\mathbb{C}),$ the corresponding set of vectors $\{|v_{i}^{\alpha }\>:\alpha \in \Theta$, $
	i=1,\ldots,|\varphi^{\alpha} |\}\subset \mathbb{C}^{n^{2}}$ defined in~\Cref{dprop14} form an orthonormal basis of the linear space $\mathbb{C}^{n^{2}}$, i.e.,
	\be
	\<v_{i}^{\alpha }|v_{j}^{\beta }\>=\delta ^{\alpha \beta }\delta _{ij}.
	\ee
\end{lemma}

\begin{proposition}
	\label{prop14} Let $\Phi \in \End\left[ \M(n, \cC)\right]$ be an ICLM which is not necessarily trace-preserving (see \cref{eq21}).
	 Then $\forall \alpha \in\Theta$, the Choi-Jamio{\l}koski images of the operators $\Pi^{\alpha}$ satisfy
	\be
	J(\Pi^\alpha)|v_{i}^{\beta }\>=\mu_i(\alpha,\beta)|v_{i}^{\beta }\>.
	\ee
	The vectors $|v_{i}^{\beta }\>$ are common eigenvectors of all $J(\Pi^{\alpha})$ with eigenvalues $\mu_i(\alpha,\beta)$ given by
	\begin{equation}\label{egval}
	\mu_{i}(\alpha ,\beta )=\frac{\left|\varphi^{\alpha} \right|}{|G|}\sum_{g\in G}\chi
	^{\alpha }\left( g^{-1}\right) \left\vert \tr\left( V_{i}^{\beta }
	U^{\dagger }(g)\right)\right\vert ^{2}.
	\end{equation}
\end{proposition}

	\begin{remark}
		\label{Rprop14}
	The vectors defined in~\cref{malutkie} of~\Cref{dprop14} are simultaneous eigenvectors of  $J(\Pi^\alpha)$, $\forall\,\alpha \in \Theta$, so
	\begin{align}\label{CJ-commute}
	\forall\,\alpha, \beta \in \Theta, \quad [J(\Pi^\alpha), J(\Pi^\beta)] = 0.
	\end{align}	
	\end{remark}

\begin{remark}
	\label{independence}
	From the structure of
	the right hand side  of~\cref{egval}, more precisely, from the presence of the modulus in it, and from~\cref{astt} in~\Cref{propP}, it follows that the eigenvalues $\mu _{i}(\alpha ,\beta )$ do not depend
	on the particular choice of the pair $(s,t) \in \cS_{\beta,i}$
used to parametrize $V_i^\beta \equiv V_i^\beta(s,t)$.
\end{remark}

From~\Cref{dprop14},~\Cref{vec_v} and~\Cref{prop14} we obtain the following corollary:
\begin{corollary}
	\label{normal}
	The operators $J(\Pi ^{\alpha })\in \M(n^{2},\mathbb{C})$, $\alpha \in \Theta $ are normal.
\end{corollary}

Using~\cref{11} of~\Cref{propP}, the eigenvalues $\mu_i(\alpha,\beta)$ (given in~\Cref{prop14}) can also be expressed equivalently as follows:
\begin{remark}
For any $(s,t)$ in the set $\cS_{\beta,i}$, defined through~\cref{setS},
	\begin{align}
\mu _{i}(\alpha,\beta )=\frac{1}{\left( \widetilde{\Pi }_{i}^{\beta }\right) _{st,st}}%
\frac{|\varphi^{\alpha} ||\varphi^{\beta} |}{|G|^{2}}\sum_{g,h\in G}\chi ^{\alpha
}\left( g^{-1}\right) \varphi _{ii}^{\beta }\left( h^{-1}\right) u_{ts}\left( g^{-1}\right) u_{st}\left( hgh^{-1}\right) ,\qquad
\alpha ,\beta \in \Theta.
	\label{rhs}
	\end{align}
	From \Cref{independence} it follows that the right hand side of~\cref{rhs}, does not depend on the particular choice of the pair $(s,t) \in \cS_{\beta,i}$.
\end{remark}

By \Cref{prop14}, the Choi-Jamio{\l}kowski image $J(\Phi )$
given by~\cref{CJ-decomp} is
a linear combination of mutually {\em{commuting}} matrices whose eigenvalues are known (i.e.~given by \cref{egval}). Hence,
we can explicitly write down the conditions for positive-semidefiniteness of $J(\Phi )$ as follows:
\begin{corollary}
	\label{cor15} The Choi-Jamio{\l}kowski image $J(\Phi )$ given by~\cref{CJ-decomp}
	is positive semi-definite if and only if its eigenvalues, which we denote by $\epsilon_i^{\beta}$, are non-negative, i.e.~for any $\beta \in \Theta $, and $i=1,\ldots ,|\varphi^{\beta} |$ 
	\begin{equation}
	\label{pos2}
	\epsilon_i^{\beta}\equiv \sum_{\alpha \in \Theta }l_{\alpha }\mu_{i}(\alpha ,\beta )=\frac{1}{|G|%
	}\sum_{g\in G}\left( \sum_{\alpha \in \Theta }l_{\alpha }\left|\varphi^{\alpha} \right| \chi
	^{\alpha }\left( g^{-1}\right) \right) \left\vert \tr\left( V_{i}^{\beta
	} U^{\dagger }(g)\right)\right\vert ^{2}\geq 0,
	\end{equation}
	where $V_i^{\beta} \in \M(n,\mathbb{C})$ is the normalized eigenvector (see~\Cref{propP}) of the projector $\Pi_i^{\beta}$. 
\end{corollary}

	\begin{remark}
	From~\cref{t-Choi} of~\Cref{Ch-slad} and~\cref{pos2} of~\Cref{cor15} it follows, that
	\be
	\label{sumE}
	\sum_{\beta \in \Theta } \ \sum_{i=1,\ldots,|\beta |}\epsilon _{i}^{\beta }=|U|.
	\ee
\end{remark}
\noindent
Using~\cref{11} of~\Cref{propP}, we get the following equivalent expression for the eigenvalues $\epsilon_i^{\beta}$.
\begin{lemma}
	\label{r28}
	For some $k, l \in \left\{1,\ldots ,n\right\}$ for which $\left( \tPi_{i}^{\beta }\right)
	_{st,st}\neq 0$, we have 
	\begin{equation}
	\label{eq281}
	\epsilon_i^{\beta}=\frac{1}{|G|^2%
	}\sum_{g\in G}\sum_{\alpha \in \Theta }l_{\alpha }\left|\varphi^{\alpha} \right| \chi ^{\alpha
}\left( g^{-1}\right) \left\vert \frac{\left|\varphi^{\beta} \right|}{\sqrt{%
	\left( \tPi_{i}^{\beta }\right) _{st,st}}}\sum_{h\in G}\varphi _{ii}^{\beta
}\left( h^{-1}\right) u_{ts}\left( h^{-1}g^{-1}h\right) \right\vert ^{2}\geq
0,
\end{equation}%
or more explicitly
\be
\label{eq282}
\begin{split}
&\epsilon _{i}^{\beta }=\sum_{\alpha \in \Theta }l_{\alpha }\mu _{i}(\alpha
,\beta )=\\
&=\frac{1}{\left( \widetilde{\Pi }_{i}^{\beta }\right) _{st,st}}\frac{|\varphi^{\beta} |}{%
	|G|^{2}}\sum_{g,h\in G}\left[ \sum_{\alpha \in \Theta }l_{\alpha }|\varphi^{\alpha} |\chi
^{\alpha }\left( g^{-1}\right) \right] \varphi _{ii}^{\beta
}\left( h^{-1}\right) u_{ts}\left( g^{-1}\right) u_{st}\left( hgh^{-1}\right)\geq 0 , 
\end{split}
\ee
for $\beta \in \Theta
, \ i=1,\ldots,\left|\varphi^{\beta} \right|$. 
\end{lemma}

	\begin{remark}
	From~\Cref{propP} and~\Cref{RpropP} it follows that the right hand side of 
	the above equations do not depend on the particular choice of the pair $s,t \in \cS_{\beta,i}$.
	\end{remark}

The properties of the eigenvalues $\mu_{i}(\alpha ,\beta )$ depend on
the properties of the particular irrep (of the group $G$) which is considered, as discussed below.

\begin{lemma}[Hermiticity of the Choi-Jamio{\l}kowski image]
	\label{hem}
	The Choi-Jamio{\l }kowski image $J(\Phi)$ of an ICLM $\Phi\in \End
	\left( \M(n,\mathbb{C})\right) $ is Hermitian if and only if
	\begin{equation}
	\forall X\in \M(n,\mathbb{C}), \quad \Phi(X)^{\dagger }=\Phi(X^{\dagger }).
	\end{equation}
	\label{Herm-Choi}
\end{lemma}

Note that for any $X\in \M(n,\mathbb{C})$, we have
\begin{equation}
\Pi^{\alpha }(X)^{\dagger}=\frac{\left|\varphi^{\alpha} \right|}{|G|}\sum_{g\in G}%
\overline{ \chi ^{\alpha }\left( g^{-1}\right) }U(g)X^{\dagger}U\left(
g^{-1}\right)
\end{equation}
and 
\begin{equation}
\Pi^{\alpha }\left( X^{\dagger}\right) =\frac{\left|\varphi^{\alpha} \right|}{|G|}%
\sum_{g\in G}\chi ^{\alpha }\left( g^{-1}\right) U(g)X^{\dagger}U\left(
g^{-1}\right) .
\end{equation}
\medskip

From the above, we can infer that if the irreducible characters of the group $G$ are real then the
	Choi-Jamio{\l }kowski image $J\left(\Pi^{\alpha }\right)$ of any projector $%
	\Pi^{\alpha}$, appearing in~\Cref{cor10}, is Hermitian and hence the corresponding eigenvalues $\mu
	_{i}(\alpha,\beta)$ are real. For example, all characters for the symmetric group $S(n)$ and quaternion group $Q$ are given by real numbers, so in this case eigenvalues $\mu
	_{i}(\alpha,\beta)$ are always real.

\section{Main Results}
\label{sec:main}
The results stated in the previous section (and proved in \Cref{appB}) are summarized in the following theorem (\Cref{thm16}). The latter gives an explicit description of an irreducibly covariant quantum channel (ICQC), corresponding to an irrep $U$ of a finite group for which $U\ot U^c$ is simply reducible. In addition, in~\Cref{KK1} we obtain explicit expressions for the Kraus operators of any ICQC. These theorems are based on the following assumption:
\begin{assum}
	\label{assumption1}
	Suppose that a unitary irrep $U:G\rightarrow \M(n,\mathbb{C})
	$ (of a finite group $G$) is such that $U\otimes U^{c}$ is simply reducible, i.e.,
	\begin{equation}
	U\otimes U^{c}=\bigoplus_{\alpha \in \Theta }\varphi^{\alpha },
	\end{equation}%
	where $\Theta $ is the index set of those irreps of $G$ which
	appear in the above decomposition.
\end{assum}

\begin{theorem}
\label{thm16}  Under~\Cref{assumption1} a linear map 
$\Phi \in \End\left[\M( n,\mathbb{C})\right] $, is an ICQC with respect to 
the irrep $U$ if and only if it has a decomposition of the following form: 
\begin{equation}
\label{comm1}
\Phi =l_{\id}\Pi^{\id}+\sum_{\alpha \in \Theta, \alpha \neq \id}l_{\alpha }\Pi^{\alpha }\quad \text{with} \quad l_{\id}=1, \ l_{\alpha}\in \mathbb{C}; \quad \Pi^{\id}, \Pi^{\alpha } \in \End\left[\M( n,\mathbb{C})\right],
\end{equation}%
where $\Pi^{\id}, \Pi^{\alpha }$ are the projectors are defined through~\cref{Pi_a} and~\cref{properties1}; the coefficients $l_{\alpha }$ are eigenvalues of $\Phi$ and satisfy the following inequalities: 
\begin{equation}
\label{sol}
\sum_{g\in G}\left( \sum_{\alpha \in \Theta }l_{\alpha }\left|\varphi^{\alpha} \right| \chi
^{\alpha }\left( g^{-1}\right) \right) \left\vert \tr\left( V_{i}^{\beta
}U^{\dagger }(g)\right) \right\vert ^{2}\geq 0, \quad \forall \beta \in \Theta, \quad i \in \{1,\ldots, |\varphi^\beta|\}.
\end{equation}%
In the above, $V_i^\beta \in \M(n, \cC)$ denote the normalized eigenvectors of rank-one projectors $\Pi_i^\beta \in \End\left[\M( n,\mathbb{C})\right]$ such that
$\Pi^\beta = \sum_i \Pi_i^\beta$, and are explicitly given in~\cref{11}.
\end{theorem}

The eigenvalue equation for the Choi-Jamio{\l}kowski image $J(\Phi)$, defined through \cref{71}, reads:
\begin{align}
J(\Phi)|v_{i}^{\beta }\>  &= \epsilon_{i}^{\beta } |v_{i}^{\beta }\>,
\end{align}
with the eigenvalues $\epsilon_{i}^{\beta }$ being given by \Cref{cor15}. From \Cref{thm16} and~\cref{11} of~\Cref{propP} it follows that for some $(s, t)$ in the set $\cS_{\beta,i}$ defined through \cref{setS}:
\begin{align}
\epsilon _{i}^{\beta } &=\frac{1}{\left( \widetilde{\Pi }_{i}^{\beta }\right) _{st,st}}\frac{|\varphi^{\beta} |}{%
		|G|^{2}}\sum_{g,h\in G}\left[ \sum_{\alpha \in \Theta }l_{\alpha }|\varphi^{\alpha} |\chi
	^{\alpha }\left( g^{-1}\right) \right] \varphi _{ii}^{\beta
	}\left( h^{-1}\right) u_{ts}\left( g^{-1}\right) u_{st}\left( hgh^{-1}\right)\geq 0 , 
\end{align}
for $\beta \in \Theta
, \ i=1,\ldots,\left|\varphi^{\alpha} \right|$. From~\Cref{independence} it follows, that the above expression for $\epsilon _{i}^{\beta }$
does not depend on the particular choice of the pair $(s,t) \in \cS_{\beta,i}$.

Using the statement of~\Cref{Kraus} of~\Cref{appC}, and~\Cref{prop14}, we are able to give the Kraus decomposition of any ICQC which satisfies ~\Cref{assumption1}. This is stated in the following theorem:
\begin{theorem}
	\label{KK1}
	The Kraus operators of any ICQC $\Phi \in \End\left[\M(n,\mathbb{C})\right] $, which satisfy ~\Cref{assumption1}, have the following form:
	\be
	\label{krausiki}
	K_{i}(\beta )=\sqrt{\epsilon _{i}^{\beta }}\left( V_{i}^{\beta }\right) ^{T},\qquad \beta
	\in \Theta ,\quad i=1,\ldots,|\Theta |,
	\ee
	where $\epsilon _{i}^{\beta }$ are eigenvalues of the Choi-Jamio{\l}kowski
	image $J(\Phi )$ given by~\cref{pos2} and $	V_{i}(\beta )=\left(\cV^{-1}\left[ |v_{i}^{\beta }\>\right] \right) ^{T} \in \M(n,\mathbb{C})$
	with $
	|v_{i}^{\beta }\>$ defined through~\cref{malutkie}.
\end{theorem}

\begin{remark}
	The matrices $K_{i}(\beta )$ given through~\cref{krausiki} depend
	on the indices $(s,t)$ in $V_{i}^{\beta }\equiv V_{i}^{\beta }(s,t)$ (see~\Cref{propP}) but from~\cref{VHS} it follows that the Kraus representation of ICQC 
	$\Phi$:
	\be
	\Phi (X)=\sum_{\beta \in \Theta}\sum_{i=1}^{\left| \varphi^{\beta}\right| }K_i(\beta )XK_{i}^{\dagger}(\beta ),\qquad X\in \M\left( n,\mathbb{C}\right) 
	\ee
	does not depend on the choice of the indices $(s,t)$ in $V_{i}^{\beta
	}\equiv V_{i}^{\beta }(s,t)$.
\end{remark}

\section{Some Geometrical Properties of ICLMs and ICQCs}
\label{geometry}
The system of linear equations (see~\cref{pos2} in~\Cref{cor15}):
\be
\epsilon _{i}^{\beta }=\sum_{\alpha \in \Theta }l_{\alpha }\mu _{i}(\alpha
,\beta ),\qquad \beta \in \Theta , \ i=1,\ldots,|\varphi^{\beta} |, 
\ee
without the assumption $l_{\id}=1$, describes a linear dependence between the vectors $L=(l_{\alpha })\in 
\mathbb{C}^{|\Theta |}$ of pairwise distinct eigenvalues\footnote{For each $\alpha \in \Theta$ there is an eigenvalue $l_\alpha$. These have multiplicity greater than one whenever $\Pi^\alpha$ is not a rank-one projector. However, these multiplicities are not taken into account in defining the vector $L$.} of the ICLM $\Phi
=\sum_{\alpha \in \Theta }l_{\alpha }\Pi ^{\alpha }$, and the vectors $E
=(\epsilon _{i}^{\beta })\in \mathbb{C}^{n^{2}}$ describing all the eigenvalues (including multiplicities) of
the Choi-Jamio{\l}kowski image $J(\Phi)$. This system of linear equations may also be written in the matrix form:
\be
\label{sets}
E =ML,\qquad M=(m_{\beta i,\alpha })\equiv \left( \mu _{i}(\alpha ,\beta
)\right)\in \M\left( n^{2}\times |\Theta |,\mathbb{C}\right) . 
\ee
In the above, we choose a particular ordering of the irreps, such that the first column and row of the matrix $M$ is always labelled by the parameter ($\id$) of the identity irrep. The matrix $M=(m_{\beta i,\alpha })\equiv \left( \mu _{i}(\alpha ,\beta )\right) $ has
the following important property:

\begin{proposition}
	\label{Mrank}
	The matrix $M=(m_{\beta i,\alpha })\equiv \left( \mu _{i}(\alpha
	,\beta )\right) $ has maximal possible rank, equal to $|\Theta |$, which means that
	the columns $M_{C(\alpha)}\in \mathbb{C}^{n^{2}}$ of the matrix $M$ are linearly independent. This implies that the
	matrix $\ M$ is invertible from the left, and denoting the left inverse as $M^{inv}$, we have:
	\be
	M^{inv}M=\text{\noindent
		\(\mathds{1}\)}_{|\Theta |}: \ M^{inv}=G^{-1}M^{\dagger}\in \M\left( |\Theta |\times
	n^{2},\mathbb{C}\right) , 
	\ee
	where $G=\left( (M_{C(\alpha)},M_{C(\alpha ^{\prime })}\right) \in \M\left(|\Theta
	|,\mathbb{C}\right) $ is the Gram matrix of the columns $M_{C(\alpha)}\in \mathbb{C}^{n^{2}}$ of the matrix $M$.
\end{proposition}	

\begin{proof}
	The Choi-Jamio{\l}kowski map $J$ when restricted to the linear subspace of the  ICLM (see~\Cref{R-ICLM}) is also an isomorphism. Any ICLM 
	$\Phi =\sum_{\alpha \in \Theta }l_{\alpha }\Pi ^{\alpha } \in \End\left[\M(n,\mathbb{C})\right]$ is, by construction, normal; it means that $\Phi =0$ if and only if all its eigenvalues $\{l_{\alpha }:\alpha \in \Theta \}$ are zero. If the rank of the matrix $M$ was smaller then $|\Theta |$, then it would mean that 
the Choi-Jamio{\l}kowski image $J(\Phi)$ of a 
nonzero ICLM $\Phi =\sum_{\alpha \in \Theta }l_{\alpha }\Pi ^{\alpha }$ (with nonzero eigenvalues $L=(l_{\alpha })$), would have all eigenvalues $E =(\epsilon _{i}^{\beta })$ equal to zero. On the other hand, the matrix $J(\Phi )$, which is a linear combination of commuting normal
	matrices $J(\Pi ^{\alpha })$, is also normal (see~\Cref{normal}), so it would mean that $J(\Phi )$ is zero, which is impossible because $J$ defines an isomorphism and therefore cannot take the value zero on a nonzero argument. The maximal rank of the matrix $M=(m_{\beta
		i,\alpha })\equiv \left( \mu _{i}(\alpha ,\beta )\right) $ implies that it is invertible. In fact, we have 
	\be
	M^{\dagger}M=G\in \M\left(|\Theta |,\mathbb{C}\right) : \ G=\left( M_{C(\alpha)},M_{C(\alpha ^{\prime })}\right) ,\quad \alpha,\alpha'\in\Theta 
	\ee
	i.e.~the matrix $G$ is the Gram matrix of the columns $M_{C(\alpha)}\in \mathbb{C}^{n^{2}}$ of the matrix $M$. From the linear independence of these columns
	it follows that the Gram matrix $G$ is invertible, so we get
	\be
	G^{-1}M^{\dagger}M=\text{\noindent
		\(\mathds{1}\)}_{|\Theta |},
	\ee
where $\text{\noindent\(\mathds{1}\)}_{|\Theta |}$ denotes the $|\Theta| \times |\Theta|$ identity matrix.
\end{proof}

\begin{remark}
	The first column of the matrix $M$ given in~\cref{sets}, $M_{C(\id)}=\left( \mu _{i}(\id,\beta )\right) $, labelled by $\alpha =\id$, is of the form
	\be
	\label{col}
	(M_{C(\id)})^{T}=\frac{1}{|U|}(1,1,\ldots,1)\in \mathbb{C}^{n^{2}},
	\ee
	and the first row $M_{R(\id)}$, indexed by $\beta=\id$ is given by:
	\be
	\label{row}
	M_{R(\id)}=\left( m_{\id,\alpha }\right) =\left( \mu (\alpha ,\id)\right) =\left( \frac{|\varphi^{\alpha}| }{|U|}\right) .
	\ee
\end{remark}

\begin{corollary}
	\label{cor43}
	The correspondence between pairwise distinct eigenvalues $L=(l_{\alpha })\in \mathbb{C}^{|\Theta |}$ of the ICLM $\Phi =\sum_{\alpha \in \Theta }l_{\alpha }\Pi
	^{\alpha }$  and eigenvalues $E =(\epsilon _{i}^{\beta })\in \mathbb{C}^{n^{2}}$ of its Choi-Jamio{\l}kowski image $J(\Phi )$ is one-to-one. The
	linear space $\mathbb{C}^{|\Theta |}$ of all possible pairwise distinct eigenvalues $L=(l_{\alpha })$
	of an ICLM is transformed isomorphically to a subspace $\mathcal{E}$ of
	dimension $|\Theta |$ in $\mathbb{C}^{n^{2}}$, spanned by vectors $E =ML$ in the linear space $\mathbb{C}^{n^{2}}$. It means that we have $\mathcal{E}=M(\mathbb{C}^{|\Theta |})$, where $M(\mathbb{C}^{|\Theta |})$ is the subspace generated by the action of the matrix $M$ on the vectors in $\mathbb{C}^{|\Theta|}$. Moreover, the condition $E =(\epsilon _{i}^{\beta })\in \mathcal{E}$ implies that 
	\be
	M^{inv}E =L=(l_{\alpha })\in \mathbb{C}^{|\Theta |}.
	\ee
\end{corollary}

Now the conditions of trace preservation and complete positivity, which an ICLM 
must
satify in order to be an ICQC, imply the following conditions on the
eigenvalues $E =(\epsilon_{i}^{\beta })$: 
\be
\label{xxx}
\epsilon _{i}^{\beta }\geq 0,\quad \beta \in \Theta, \ i=1,\ldots,|\varphi^{\beta}|,\qquad
\sum_{\beta \in \Theta } \ \sum_{i=1,\ldots,|\varphi^{\beta}|}\epsilon _{i}^{\beta }=|U|. 
\ee
Let us first recall the definition of a simplex. 
\begin{definition}
	\label{Sym}
	The set, 
	\be
	\left\lbrace (x_{1},x_{2},...,x_{d+1})\in 
	\mathbb{R}^{d+1}:x_{i}\geq 0, \quad  \sum_{i=1}^{d+1}x_{i}=1\right\rbrace 
	\ee
	is called a {\em{standard $d$-dimensional simplex}}.
\end{definition}

Comparing~\cref{xxx}  with~\Cref{Sym} of the simplex we see that the conditions on the eigenvalues $\epsilon _{i}^{\beta },$ with $\beta \in \Theta , \ i=1,\ldots,|\varphi^{\beta}|$, coincide with the conditions of  a
simplex scaled by a factor $|U|$. Let us denote such a scaled simplex by $\Sigma(U)$, so we have 
\be
\label{contracted}
\Sigma(U)=\left\lbrace (x_{1},x_{2},\ldots,x_{n^{2}})\in 
\mathbb{R}^{n^{2}}:x_{i}\geq 0,\quad \sum_{i=1}^{n^{2}}x_{i}=|U|\right\rbrace \subset \mathbb{C}^{n^{2}}.
\ee
From~\cref{contracted} and the definition of subspace $\mathcal{E}$ given in~\Cref{cor43} we get the following:

\begin{proposition}
	\label{intersection}
	The eigenvalues $E =( \epsilon _{i}^{\beta }) $ of the
	Choi-Jamio{\l}kowski image $J(\Phi )$ of an ICQC $\Phi =\sum_{\alpha \in \Theta }l_{\alpha }\Pi ^{\alpha }$ lie  in the intersection of the scaled  simplex $
	\Sigma(U)\subset \mathbb{C}^{n^{2}}$ and the subspace $\mathcal{E}\subset \mathbb{C}^{n^2}$, i.e.,
	\be
	E =(\epsilon _{i}^{\beta })\in \Sigma(U)\cap \mathcal{E}.
	\ee
\end{proposition}

\bigskip From this statement and~\Cref{cor43} we get the following
characteristic  of the $ICQC$ $\Phi =\sum_{\alpha \in \Theta }l_{\alpha }\Pi
^{\alpha }$.
\begin{corollary}
	\label{cor36}
	An ICLM $\Phi =\sum_{\alpha \in \Theta }l_{\alpha }\Pi ^{\alpha }$ is an ICQC if and only if its vector of eigenvalues $L=(l_{\alpha })\in \mathbb{C}^{|\Theta |}$ is an inverse image, in the mapping $M:\mathbb{C}^{|\Theta |}\rightarrow \mathcal{E}$ $=M\left( \mathbb{C}^{|\Theta |}\right) $, of some vector $E =(\epsilon _{i}^{\beta })\in 
	\Sigma(U)\cap \mathcal{E}$, i.e.
	\be
	L=(l_{\alpha })=M^{inv}(E ),
	\ee
	where $M^{inv}$ is the left inverse of the linear mapping $M:\mathbb{C}^{|\Theta |}\rightarrow \mathcal{E}=M(\mathbb{C}^{|\Theta |})$ described in~\Cref{Mrank}.
\end{corollary}

\begin{remark}
	In particular one can check by direct computation, that the
	structure of the matrix $M=\left( m_{\beta i,\alpha }\right) \equiv \left(\mu _{i}(\alpha,\beta )\right) $ is such that, $\sum_{\beta \in \Theta }\sum_{i=1,\ldots,|\varphi^{\beta}|}\epsilon_{i}^{\beta }=|U|$ if and  only if $l_{\id}=1$, i.e.~the ICLM $
	\Phi =\sum_{\alpha \in \Theta }l_{\alpha }\Pi ^{\alpha }$ is trace
	preserving.
\end{remark}

\section{Examples of ICQC}
In this section we give an explicit choice of the parameters $l_\alpha$ (occurring in \cref{comm1}), for which the inequality \cref{sol}, which gives the condition for complete positivity of an ICLM, is automatically satisfied. For this choice, the ICLM given by \cref{comm1} is therefore completely positive, and hence defines an ICQC if in addition $l_{\rm{id}} =1$.

In addition, we provide explicit examples of ICQCs for some fixed groups for which tensor product $U\ot U^c$  is simply reducible for an irrep $U$. We focus on the quaternion group $Q$ and the symmetric groups $S(3)$ and $S(4)$. In each case we present both the matrix representation and the Kraus representation of the ICQC. Further, for the case of $S(3)$ and $Q$, using the Peres-Horodecki (or PPT) criterion~\cite{SepHHH,SepAP}, we find the condition under which the ICQC is an entanglement breaking channel.

\subsection{A wide class of ICQC}
\label{wide_class}
Here we give a wide class of ICQCs by providing explicit expressions for the
eigenvalues $l_{\alpha}$ in the spectral decomposition of an ICLM $
\Phi =\sum_{\alpha \in \Theta }l_{\alpha }\Pi^{\alpha}$ for which $\Phi$ is completely positive and trace-preserving. These eigenvalues form a class of solutions of~\cref{sol} and they
are given in the following theorem:
\begin{theorem}
	\label{wc}
	Let $K(g)\subset G$ be the conjugacy class of $g \in G $, defined through \cref{conju}, and let $f:G\rightarrow \mathbb{C}$ be a function on the group $G$ such that:
	\begin{enumerate}
		\item 
		\be
		\label{fc1}
		\sum_{g\in G}f(g)=|G|.
		\ee
		\item For all conjugacy classes $K(g)$,
		\be
		\label{fc1a}
		\sum_{h\in K(g)}f(h)\geq 0.
		\ee
	\end{enumerate}
	Then a family of ICQC (i.e.~a family of quantum channels satysfying~\Cref{thm16}) are those for which the coefficients $l_{\alpha}$ in~\cref{comm1} are given by:
	\be
	l_{\alpha }=\frac{1}{|G|}\frac{1}{|\varphi^{\alpha}|}\sum_{g\in G}\chi
	^{\alpha }(g)f(g),\quad \alpha \in \Theta.
	\ee
\end{theorem}

\begin{remark}
	The first condition on the function $f:G\rightarrow \mathbb{C}$ given  in~\cref{fc1,fc1a} guarantees that the ICLM $\Phi$ is trace-preserving, whereas the second condition in~\cref{fc1,fc1a} imply that $\Phi $ is also completely positive.
\end{remark}

In order to prove~\Cref{wc} we first make the following trivial extension of the expression~\cref{decomp}, which is implied by~\Cref{prop8} and~\Cref{cor10}.
\begin{corollary}
	\label{C13}
An ICLM $\Phi \in \End\left[\M(n,\mathbb{C})\right] $ can be expressed as follows:
	\be
	\Phi=\sum_{\gamma \in \widehat{G}} l_{\gamma}\Pi^{\gamma}=\sum_{\alpha \in \Theta}l_{\alpha}\Pi^{\alpha}+\sum_{\gamma \notin \Theta} l_{\gamma}\Pi^{\gamma} \ :\quad l_{\alpha},l_{\gamma}\in \mathbb{C},
	\ee
	where $\widehat{G}$ is the set of all irreps of the group $G$, and where the $l_{\gamma}$, with $\gamma \notin \Theta$, are arbitrary. 
\end{corollary}
 From the statement of~\Cref{C13}  it follows that the
Choi-Jamio{\l}kowski image of the ICQC $\Phi $ may be written also
using the sum over all $\gamma \in \widehat{G}$:
\be
\begin{split}
J\left( \Phi \right) &=J\left( \sum_{\gamma \in \widehat{G}}l_{\gamma }\Pi^{\gamma }\right) =\frac{1}{|G|}
\sum_{ij}E_{ij}\otimes \sum_{g\in G}\left( \sum_{\gamma \in \widehat{G}}l_{\gamma
}|\varphi^{\gamma}|\chi ^{\gamma
}\left( g^{-1}\right) \right) U_{C(i)}(g)U_{R(j)}\left( g^{-1}\right),
\end{split} 
\ee
with $l_{\id}=1$. Obviously $J(\Pi^{\gamma })=0$ if $\gamma \notin \Theta $, and one can
check by a direct calculation that the corresponding eigenvalues vanish, i.e.
\be
\mu_{i}(\gamma ,\beta )=\frac{|\varphi^{\gamma}|}{|G|}\sum_{g\in
	G}\chi ^{\gamma }\left( g^{-1}\right) \left| \tr\left( V_{i}^{\beta }U^{\dagger}(g)\right) \right| ^{2}=0,\quad \quad
\beta \in \Theta , \ i=1,\ldots,|\varphi^{\beta}|. 
\ee
Finally, the inequalities given in~\cref{sol} of~\Cref{thm16} may be equivalently formulated
as follows:
\be
\label{92a}
\sum_{g\in G}\left( \sum_{\gamma \in \widehat{G}}l_{\gamma }|\varphi^{\gamma}|
\chi ^{\gamma }\left( g^{-1}\right) \right) \left| \tr\left( V_{i}^{\beta }U^{\dagger}(g)\right) \right| ^{2}=\sum_{g\in
G}x(g)\left| \tr\left( V_{i}^{\beta }U^{\dagger}(g)\right) \right| ^{2}\geq 0,
\ee
where $\beta \in \Theta, \ i=1,\ldots,|\varphi^{\beta}|$, and 
$x(g)\equiv \sum_{\gamma \in \widehat{G}}l_{\gamma }|\varphi^{\gamma}|\chi ^{\gamma }\left( g^{-1}\right)$. This form of the inequalities
is more convenient to deal with. In particular, the form of the
coefficients $x(g)$ allows us to apply the
orthogonality relations for irreps. Now we focus our attention
on these coefficients. The idea of the construction of the solution of
the inequalities from~\cref{92a} is to find $\{l_{\gamma }:$ $\gamma \in \widehat{G}\}$
such that the coefficients $x(g)$ are non-negative for any $g\in G$, which
obviously implies that the inequalities from~\cref{92a} are satisfied. In order to
study the properties of the coefficients $x(g)$  let us introduce the
following matrices:

\begin{definition}
	\label{AD1}
	Let $T=(t_{g\gamma })\equiv \left( \chi ^{\gamma }\left( g^{-1}\right) \right) \in \M\left( |G|\times |\widehat{G}|,\mathbb{C}\right)$, where $g\in G$ and $\gamma \in \widehat{G}$ be a
	rectangular matrix whose rows are indexed by the group elements $g\in G$ and
	columns are indexed by irreps $\gamma \in \widehat{G}$. We
	assume that the group elements are ordered in a way such that elements belonging to the same conjugacy class are grouped together. Let $D\in \M(|\widehat{G%
	}|,\mathbb{C}) $ be a square matrix defined as follows: 
	\be
	D=\left( d_{\gamma \delta })\equiv \operatorname{diag}(|\varphi^{\gamma_1}|,|\varphi^{\gamma_2}|,\ldots,|\varphi^{\gamma_r}|\right)  \ : \ \gamma _{i}\in \widehat{G},
	\ee
	and let us define the following set of vectors: 
	\be
	\widehat{L}\equiv(l_{\gamma })\in 	\mathbb{C}^{\left| \widehat{G}\right| },\qquad F\equiv\left( f(g)\right) \in \mathbb{C}^{\left| G\right| }.
	\ee
\end{definition}

\bigskip Using Schur's orthogonatlity relations for irreducible characters given
in~\cref{or} and~\cref{or2}, one can show that the matrix $T$ satisfies the
relations given in the following lemma. They can be verified by direct computations. It is well-known that for finite
	groups the number of inequivalent irreps is equal to the number
	of conjugacy classes~\cite{FHa}. This allows us to label the congugacy classes $K(g)$, $g \in G$ with the indices $\gamma \in \widehat{G}$ used for the irreps. Hence, we interchangeably denote the conjugacy classes as $K(g)$ or $K_\gamma$.
\begin{lemma}
	\label{lemma36}
	Suppose that we are given with the matrix $T=(t_{g\gamma })=\left( \chi ^{\gamma }\left( g^{-1}\right) \right) \in \M\left( |G|\times |\widehat{G}|,\mathbb{C}\right) $, where the group elements $g\in G$  are ordered in such a way that the
	conjugated elements of the group are all grouped together. Then,
	\be
	\label{sec}
	\frac{1}{|G|}T^{\dagger}T=\text{\noindent
		\(\mathds{1}\)}_{\left| \widehat{G}\right| }\in \M(|\widehat{G}|,\mathbb{C}),\qquad \frac{1}{|G|}TT^{\dagger}=\bigoplus _{\gamma \in \widehat{G}}\frac{1}{%
		|K_{\gamma }|}\mathbb{I}_{|K_{\gamma }|}\in \M(|G|,\mathbb{C}),
	\ee
where the matrices $\mathbb{I}_{|K_{\gamma }|}\in \M(|K_{\gamma }|,\mathbb{C})$ and all their entries are equal to $1$ and $K_{\gamma }:$ $\gamma \in 
	\widehat{G}$ are classes of conjugated elements. The matrix $\frac{1}{|G|}TT^{\dagger}$  is a block-diagonal matrix with diagonal blocks equals to the
	matrices $\mathbb{I}_{|K_{\gamma }|}$ and all remaining blocks are equal to zero.
\end{lemma}
Using these matrices we may write the set of equations $x(g)=\sum_{\gamma
	\in \widehat{G}}l_{\gamma }|\varphi^{\gamma}|\chi ^{\gamma
}\left( g^{-1}\right) :g\in G$ in the matrix form:
\be
X=(x(g))=TD\widehat{L}\Leftrightarrow x(g)=\sum_{\gamma ,\delta \in \widehat{G}%
}t_{g\gamma }d_{\gamma \delta }l_{\delta }=\sum_{\gamma \in \widehat{G}%
}l_{\gamma }|\varphi^{\gamma}|\chi ^{\gamma }\left( g^{-1}\right) \ : \ g\in G.
\ee
Now we are in the position to prove~\Cref{wc}.
\begin{proof}[Proof of~\Cref{wc}]
	We see that 
	\be
	\widehat{L}=\frac{1}{|G|}D^{-1}T^{\dagger}F\Leftrightarrow l_{\gamma }=\frac{1}{|G|}\frac{1}{|\varphi^{\gamma}| }\sum_{g\in G}\chi ^{\gamma }(g)f(g):\gamma \in \widehat{G},
	\ee
	then, using~\Cref{AD1} and the second statement of~\Cref{lemma36}, we find that the
	coefficients 
	\be
	\label{blabla}
	x(g)\equiv \sum_{\gamma \in \widehat{G}}l_{\gamma }|\varphi^{\gamma}|\chi ^{\gamma }(g^{-1}):g\in G
	\ee
	occurring in~\cref{92a} become%
	\be
	\label{93}
	\left( x(g)\right) =TD\widehat{L}=\frac{1}{|G|}TT^{\dagger}F\Rightarrow \forall g\in G \ 
	x(g)=\frac{1}{\left|K(g) \right| }\sum_{h\in K(g)}f(h).
	\ee
	From the above we see that for any $g\in G$ the coefficient $x(g)$ is non-negative by assumption on
	\ the function $f:G\rightarrow \mathbb{C}$ and therefore the inequalities given in~\cref{92a} and~\cref{sol} are satisfied because on the 
	$LHS$ of these inequalities we have the sum of non-negative numbers. Summarizing, for a
	given function $f:G\rightarrow \mathbb{C}$, which satisfies the assumptions of~\Cref{wc} we get a set of $|\widehat{G}|
	$ numbers
	\be
	l_{\gamma }=\frac{1}{|G|}\frac{1}{|\varphi^{\gamma}| }\sum_{g\in G}\chi ^{\gamma
	}(g)f(g):\gamma \in \widehat{G}
	\ee
	for which coefficients $x(g)$ from~\cref{blabla} are automatically non-negative. Since the sum in the decomposition of an ICQC $\Phi =\sum_{\alpha \in \Theta }l_{\alpha }\Pi^{\alpha}$ runs over those $\alpha$ which occur in the decomposition~\cref{eq15} of $U\ot U^c$, we can restrict ourselves to the subset $\{l_{\alpha} \ : \ \alpha \in \Theta\}$.
\end{proof}

Now we illustrate~\Cref{wc} by two examples:
\begin{example}
	If $f:G\rightarrow \mathbb{C}$ is such that $\forall g\in G$ \ $f(g)=1$, which obviously satisfies the
	conditions of~\Cref{wc}, then using~\cref{93} we obtain:
	\be
	l_{\id}=1,\quad l_{\alpha }=0,\quad \alpha \neq \id
	\ee
	and $\Phi =\Pi^{\id}$.
\end{example}

\begin{example}
	It is clear that any function $f:G\rightarrow \mathbb{C}$ such that $\forall g\in G$ \ $f(g)\geq 0$ and  $\frac{1}{|G|}\sum_{g\in
		G}f(g)=1$ which satisfies  the conditions of~\Cref{wc} defines a probability distribution. From this it follows that for any probability distribution on a finite group $G$, one can define an ICQC.
\end{example}

\subsection{Quaternion group $Q$}
\label{Quat}
The quaternion group $Q=\left\lbrace \pm Q_{\operatorname{e}},\pm Q_1, \pm Q_2, \pm Q_3 \right\rbrace $ is a non-abelian group of order eight satysfying
\be
Q=\left\langle -Q_{\operatorname{e}},Q_1,Q_2,Q_3 \ | \ \left(-Q_{\operatorname{e}} \right)^2=Q_{\operatorname{e}}, Q_1^2=Q_2^2=Q_3^2=Q_1Q_2Q_3=-Q_{\operatorname{e}} \right\rangle.
\ee
It possesses five inequivalent irreducible representations which we label by $\id,t_1,t_2,t_3,t_4$, respectively. However, only one of them, labelled by $t_4$, has dimension greater than one and its dimension is equal to two. 
It is known that the quaternion group can be represented as a subgroup of $GL(2,\mathbb{C})$. The matrix representation $R:Q\rightarrow GL(2,\mathbb{C})$ is given by
\be
Q_{\operatorname{e}}=\begin{pmatrix}
	1 & 0 \\ 0 & 1
\end{pmatrix}, \  Q_1=\begin{pmatrix}
\operatorname{i} & 0\\ 0 & -\operatorname{i}
\end{pmatrix}, \ Q_2=\begin{pmatrix}
0 & 1\\ -1 & 0
\end{pmatrix}, \ Q_3=\begin{pmatrix}
0 & \operatorname{i} \\ \operatorname{i} & 0
\end{pmatrix},
\ee
where $\operatorname{i}^2=-1$. In~\cref{t} we present values of the characters for all irreducible representaions of the group $Q$.

\begin{table}[h]
	\centering
	\begin{tabular}{|c|c|c|c|c|c|c|c|c|}
		\hline
		$Q$ & $Q_{\operatorname{e}}$ & $-Q_{\operatorname{e}}$ & $Q_1$ & $Q_2$ & $Q_3$ & $-Q_1$ & $-Q_2$ & $-Q_3$\\
		\hline
		$\chi^{\id}$ & 1 & 1 & 1 & 1 & 1 & 1 & 1 & 1\\
		\hline
		$\chi^{t_1}$ & 1 & 1 & -1 & 1 & -1 & -1 & 1 & -1\\
		\hline
		$\chi^{t_2}$ & 1 & 1 & 1 & -1 & -1 & 1 & -1 & -1\\
		\hline
		$\chi^{t_3}$ & 1 & 1 & -1 & -1 & 1 & -1 & -1 & 1\\
		\hline
		$\chi^{t_4}$ & 2 & -2 & 0 & 0 & 0 & 0 & 0 & 0\\
		\hline
	\end{tabular}
	\caption{Table of characters for the quaternion group $Q$.}
	\label{t}
\end{table}
Now we can construct an ICQC  for the quaternion group $Q$ with respect to the two-dimensional irrep $U=t_4$. In this case the decomposition given by~\cref{eq15} takes a form:
\begin{equation}
U \ot U^c=U^{\id} \oplus U^{t_1} \oplus U^{t_2}\oplus U^{t_3}, \quad \dim\left[\Int_{Q} \left(U\otimes U^c\right) \right]=4,
\end{equation}
so $\Theta=\left\lbrace \id,t_1,t_2,t_3\right\rbrace $. The matrix representation $\widetilde{\Phi}^{t_4}$ of the ICLM $\Phi^{t_4}$ (see~\Cref{cor10}) is given by the following expression:
\be
\label{uuu1}
\widetilde{\Phi}^{t_4}=l_{t_{\id}}\widetilde{\Pi}^{\id}+l_{t_1}\widetilde{\Pi}^{t_1}+l_{t_2}\widetilde{\Pi}^{t_2}+l_{t_3}\widetilde{\Pi}^{t_3}=\frac{1}{2}\begin{pmatrix}
	l_{t_{\id}}+l_{t_2} & 0 & 0 & 1-l_{t_2}\\
	0 & l_{t_1}+l_{t_3} & l_{t_3}-l_{t_1} & 0\\
	0 & l_{t_3}-l_{t_1} & l_{t_1}+l_{t_3} & 0\\
	1-l_{t_2} & 0 & 0 & l_{t_{\id}}+l_{t_2}
\end{pmatrix},
\ee 
where $l_{t_{\id}},l_{t_1},l_{t_2}, l_{t_3}\in \mathbb{R}$ (see explanation below~\Cref{hem}).  From Proposition~\ref{prop12}, we know
also that such a map $\Phi^{t_4}$ is trace-preserving if and only if $l_{t_{\id}}=1$, so its matrix representation in~\cref{uuu1} reduces to:
\be
\label{uuu}
\widetilde{\Phi}^{t_4}=\widetilde{\Pi}^{\id}+l_{t_1}\widetilde{\Pi}^{t_1}+l_{t_2}\widetilde{\Pi}^{t_2}+l_{t_3}\widetilde{\Pi}^{t_3}=\frac{1}{2}\begin{pmatrix}
	1+l_{t_2} & 0 & 0 & 1-l_{t_2}\\
	0 & l_{t_1}+l_{t_3} & l_{t_3}-l_{t_1} & 0\\
	0 & l_{t_3}-l_{t_1} & l_{t_1}+l_{t_3} & 0\\
	1-l_{t_2} & 0 & 0 & 1+l_{t_2}
\end{pmatrix}.
\ee  
A direct calculation gives the following explicit formula for the Choi-Jamio{\l}kowski image
of $\Phi^{t_4}$ given by~\cref{xx}:
\be
J\left(\Phi^{t_4} \right)=\sum_{i,j=1}^2 E_{ij}\ot \Phi^{t_4}\left(E_{ij} \right)=\frac{1}{2}\begin{pmatrix}
	1+l_{t_2} & 0 & 0 & l_{t_1}+l_{t_3}\\
	0 & 1-l_{t_2} & l_{t_3}-l_{t_1} & 0\\
	0 & l_{t_3}-l_{t_1} & 1-l_{t_2} & 0\\
	l_{t_1}+l_{t_3} & 0 & 0 & 1+l_{t_2}
\end{pmatrix}.
\ee
Further, to ensure that the trace-preserving ICLM $\Phi^{t_4}$ is completely positive (so that it is an ICQC) we require that $J\left(\Phi^{t_4} \right)\geq 0$. This yields:
\be
\label{cc1}
\arraycolsep=1.4pt\def\arraystretch{1.4}
\begin{array}{ll}
	\multirow{4}{*}{$J\left(\Phi^{t_4}\right)\geq 0 \  \Leftrightarrow \  $}& \epsilon^{t_1}=\frac{1}{2}\left(1+l_{t_1}-l_{t_2}-l_{t_3} \right)\geq 0, \\ & \epsilon^{t_2}=\frac{1}{2}\left(1-l_{t_1}+l_{t_2}-l_{t_3} \right)\geq 0, \\ & \epsilon^{t_3}=\frac{1}{2}\left(1-l_{t_1}-l_{t_2}+l_{t_3} \right)\geq 0, \\ & \epsilon^{\id}=\frac{1}{2}\left(1+l_{t_1}+l_{t_2}+l_{t_3} \right)\geq 0. 
\end{array}
\ee
The expressions for the eigenvalues of $J(\Phi^{t_4})$ from~\cref{cc1} may be written in the matrix form using geometrical approach presented in~\Cref{geometry}. Namely constructing matrix $M$ and the vectors $E$ and $L$ as in~\cref{sets} we get:
\be
\left( 
\begin{array}{c}
	\epsilon^{\id} \\ 
	\epsilon^{t_1} \\ 
	\epsilon^{t_2} \\ 
	\epsilon^{t_3}%
\end{array}%
\right) =\frac{1}{2}\left( 
\begin{array}{cccc}
	1 & 1 & 1 & 1 \\ 
	1 & 1 & -1 & -1 \\ 
	1 & -1 & 1 & -1 \\ 
	1 & -1 & -1 & 1%
\end{array}%
\right) \left( 
\begin{array}{c}
	l_{id} \\ 
	l_{t_1} \\ 
	l_{t_2} \\ 
	l_{t_3}%
\end{array}%
\right).
\ee
It easy to check that the matrix $M$ is invertible (exactly it is
orthogonal i.e.~we have $M^{T}=M$). From~\Cref{cor36} it follows that:
\be
\label{ls}
\left( 
\begin{array}{c}
	l_{\id} \\ 
	l_{t_1} \\ 
	l_{t_2} \\ 
	l_{t_3}%
\end{array}%
\right) =\frac{1}{2}\left( 
\begin{array}{cccc}
	1 & 1 & 1 & 1 \\ 
	1 & 1 & -1 & -1 \\ 
	1 & -1 & 1 & 1 \\ 
	1 & -1 & -1 & 1%
\end{array}%
\right) \left( 
\begin{array}{c}
	\epsilon^{\id} \\ 
	\epsilon^{t_1} \\ 
	\epsilon^{t_2} \\ 
	\epsilon^{t_3}%
\end{array}%
\right).
\ee
In particular, because of the explanations given below~\Cref{hem} we have $\forall \alpha\in \Theta \ l_{\alpha}\in \mathbb{R}$, and thanks to~\cref{sumE} in~\Cref{cor15} we are allowed to write:
\be
\label{lid}
l_{\id}=\frac{1}{2}\left( \epsilon^{\id}+\epsilon^{t_1}+\epsilon^{t_2}+\epsilon^{t_3}\right)=1.
\ee

Equations~\ref{ls} and~\ref{lid} describe the all possible values of eigenvalues $
L=(l_{\alpha })$ of ICLM $\Phi =\sum_{\alpha \in \Theta }l_{\alpha }\Pi
^{\alpha }$ for which $\Phi $ is trace-preserving. We see here that the case of
quaternion group $Q$ is very particular, because we have equality $|U|^{2}=4=|\Theta |$, and therfore the eigenvalues $L=(l_{\alpha })$ are
generated by all points $E=(\epsilon^{\beta})$ of the scaled simplex $\Sigma(|U|)$ (see~\cref{contracted} and~\Cref{intersection}).

Further from conditions $\epsilon^{\beta}\geq 0$ and~\cref{lid} the set of inequalities in~\cref{cc1} reduces to:
\be
|l_{\alpha}|\leq \frac{1}{2}\sum_{\beta \in \Theta}|\epsilon^{\beta}|=\frac{1}{2}%
\sum_{\beta \in \Theta}\epsilon^{\beta}=1,\quad \forall \alpha\in \Theta \setminus  \{\id\},
\ee
so all $l_{\alpha}$ where $\alpha\in \Theta \setminus \{\id\}$ are included in a three dimensional cube. 
The allowed values of the parameters $l_{t_1},l_{t_2},l_{t_3}$ satisfying the above constraints (for which $\Phi^{t_4}$ is an ICQC) are graphically presented on the left panel of~\cref{S31}.

In the next step we construct Kraus operators for the ICQC $\Phi^{t_4}$. From the general considerations given in~\Cref{prop14}, we are able to compute orthogonalized eigensystem of the Choi-Jamio{\l}kowski image $J(\Phi^{t_4})$.
Then, using~\Cref{Kraus} and its particular form given by~\Cref{KK1}, we can construct Kraus operators for the channel $\Phi^{t_4}$:
\be
\begin{split}
	K(t_1)&=\sqrt{\frac{\epsilon^{t_4}}{2}}\begin{pmatrix} 0 & 1\\ -1 & 0 \end{pmatrix}, \qquad K(t_2)=\sqrt{\frac{\epsilon^{t_4}}{2}}\begin{pmatrix}
		-1 & 0 \\ 0 & 1
	\end{pmatrix},\\
	K(t_3)&=\sqrt{\frac{\epsilon^{t_4}}{2}}\begin{pmatrix}
		0 & 1 \\
		1 & 0
	\end{pmatrix},\qquad K(\id)=\sqrt{\frac{\epsilon^{\id}}{2}}\begin{pmatrix}
	1 & 0\\
	0 & 1
\end{pmatrix}.
\end{split}
\ee
Moreover we have $\sum_{\beta \in \Theta}K(\beta)K^{\dagger}(\beta)=\text{\noindent
	\(\mathds{1}\)}$ and $\sum_{\beta \in \Theta}K^{\dagger}(\beta)K(\beta)=\text{\noindent
	\(\mathds{1}\)}$. Hence, the channel $\Phi^{t_4}$ is unital and trace-preserving.

\subsection{Symmetric group $S(3)$}
\label{ex-S(3)}

In the case of $G=S(3)$ we have three inequivalent irreducible representations (see~\cite{FHa},~\cite{JJF},~\cite{NS}), one-dimensional identity representation (denoted by $\id$), one-dimensional sign representation (denoted by $\sgn$) and two-dimensional nontrivial representation (denoted by $\lambda$). Here we construct a trace-preserving ICLM and in particular ICQC for the irrep $U$ characterised by the  $
\lambda$ by using so called $\epsilon-$representation \cite{NS}. The generators of this representation are given by:
\be
\label{gen}
\varphi^{\lambda}(12)=\begin{pmatrix}
	0 & 1\\ 1 & 0
\end{pmatrix},\quad \varphi^{\lambda}(23)=\begin{pmatrix}
0 & w^2\\ w & 0
\end{pmatrix},
\ee
where $w=\operatorname{exp}\left(\frac{2\pi \operatorname{i}}{3} \right) $, and by $(12),(23)$ we denote transpositions between respective elements. Decomposition given by~\cref{eq15} in this case takes a form:
\be
U\otimes U^{c}=U^{\id}\oplus U^{\sgn}\oplus U,\quad \text{and}\quad \dim \left[ \Int_{S(3)}\left( U\otimes
U^{c}\right)\right] =3. 
\ee
We see that all possible irreps occour in the decomposition of $U\ot U^c$, so $\Theta =\left\lbrace \id, \sgn, \lambda\right\rbrace $. The matrix representation $\widetilde{\Phi}_{\epsilon}$ of the trace-preserving ICLM $\Phi_{\epsilon}$ (see~\Cref{prop12}) is given by the following expression:
\be
\label{xx}
\widetilde{\Phi}_{\epsilon} =\widetilde{\Pi}^{\id}+l_{\sgn}\widetilde{\Pi}^{\sgn}+l_{\lambda }\widetilde{\Pi}^{\lambda}=\frac{1}{2}\left( 
\begin{array}{cccc}
	1+l_{\sgn} & 0 & 0 & 1-l_{\sgn} \\ 
	0 & 2l_{\lambda } & 0 & 0 \\ 
	0 & 0 & 2l_{\lambda } & 0 \\ 
	1-l_{\sgn} & 0 & 0 & 1+l_{\sgn}
\end{array}
\right). 
\ee
where $l_{\sgn},l_{\lambda}\in \mathbb{R}$ (see explanation below~\Cref{hem}). The corresponding Choi-Jamio{\l}kowski image is given by:
\be
\label{41}
J\left( \Phi_{\epsilon} \right) =\sum_{i,j=1}^2 E_{ij}\otimes \Phi_{\epsilon}\left(E_{ij} \right)=\left( 
\begin{array}{cccc}
	\frac{1}{2}(1+l_{\sgn}) & 0 & 0 & l_{\lambda } \\ 
	0 & \frac{1}{2}(1-l_{\sgn}) & 0 & 0 \\ 
	0 & 0 & \frac{1}{2}(1-l_{\sgn}) & 0 \\ 
	l_{\lambda } & 0 & 0 & \frac{1}{2}(1+l_{\sgn})
\end{array}
\right).
\ee
Similarly as in~\Cref{Quat} we can use results from~\Cref{geometry} and construct set of liniear constraints $E=ML$:
\be
\left( 
\begin{array}{c}
	\epsilon^{\id} \\ 
	\epsilon^{\sgn} \\ 
	\epsilon _{1}^{\lambda} \\ 
	\epsilon _{2}^{\lambda}%
\end{array}%
\right) =\frac{1}{2}\left( 
\begin{array}{ccc}
	1 & 1 & 2  \\ 
	1 & 1 & -2  \\ 
	1 & -1 & 0 \\ 
	1 & -1 & 0
\end{array}%
\right) \left( 
\begin{array}{c}
	1 \\ 
	l_{\sgn} \\ 
	l_{\lambda}
\end{array}%
\right).
\ee
This gives us the spectrum of $J(\Phi_{\epsilon})$:
\be
 \left\lbrace \epsilon_1^{\lambda}=\frac{1}{2}(1-l_{\sgn}),\epsilon_2^{\lambda}=\frac{1}{2}(1-l_{\sgn}),\epsilon^{\id}=\frac{1}{2}
(1+l_{\sgn})+l_{\lambda },\epsilon^{\sgn}=\frac{1}{2}(1+l_{\sgn})-l_{\lambda }\right\rbrace.
\ee
From~\Cref{cor36} it follows that 
\be
\left( 
\begin{array}{c}
	1 \\ 
	l_{\sgn} \\ 
	l_{\lambda}
\end{array}%
\right) =\frac{1}{2}\left( 
\begin{array}{cccc}
	1 & 1 & 1 & 1  \\ 
	1 & 1 & -1 & -1  \\ 
	1 & -1 & 0& 0 
\end{array}%
\right)\left(  \begin{array}{c}
	\epsilon^{\id} \\ 
	\epsilon^{\sgn} \\ 
	\epsilon _{1}^{\lambda} \\ 
	\epsilon _{2}^{\lambda}%
\end{array}
\right).
\ee
Conditions $\epsilon_i^{\beta}\geq 0$ and $\frac{1}{2}(\epsilon^{\id}+\epsilon^{\sgn}+\epsilon_1^{\lambda}+\epsilon_2^{\lambda})=1$  yields to the following statement:
\be
\label{ecp}
J\left( \Phi_{\epsilon} \right) \geq 0 \quad \Leftrightarrow \quad 1\geq l_{\sgn}\geq -1,\quad \frac{1}{2}
(1+l_{\sgn})\geq \left| l_{\lambda }\right| . 
\ee
Graphical representation of the constraints given by the above inequalities 
is presented on the left panel of~\Cref{S3a} in the Introduction.
The form of the matrix representation ${\widetilde{\Phi}}_{\epsilon} $ given by \cref{xx} guarantees that the linear map $\Phi_\epsilon$ is covariant with respect to
the irrep $U$ characterised by the irrep $\lambda$ and trace-preserving. Moreover, under the constraints on parameters $l_{\sgn}, l_{\lambda
}$ given by the inequalities (\ref{ecp}), $\Phi_{\epsilon} $ is completely positive map, hence an ICQC.

Finally using~\Cref{KK1} we construct the Kraus operators of the ICQC ${\Phi}_{\epsilon}$, they are of the following form:
\be
\label{expKa}
\begin{split}
	K_1(\lambda)&=\sqrt{\epsilon_1^{\lambda}}X_1^{\dagger}(\lambda)=\sqrt{\epsilon_1^{\lambda}}\begin{pmatrix}
		0 & 0\\ 1 & 0
	\end{pmatrix},\quad K_2(\lambda)=\sqrt{\epsilon_2^{\lambda}}X_2^{\dagger}(\lambda)=\sqrt{\epsilon_2^{\lambda}}\begin{pmatrix}
	0 & 1\\0 & 0
\end{pmatrix},\\
K(\sgn)&=\sqrt{\epsilon^{\sgn}}X_1^{\dagger}(\sgn)=\sqrt{\frac{\epsilon^{\sgn}}{2}}\begin{pmatrix}
	-1 & 0\\ 0 & 1
\end{pmatrix},\quad 	K(\id)=\sqrt{\epsilon^{\id}}X_1^{\dagger}(\id)=\sqrt{\frac{\epsilon^{\id}}{2}}\begin{pmatrix}
1 & 0\\ 0 & 1\end{pmatrix}.
\end{split}
\ee
The channel is easily checked to be unital.

\subsection{Symmetric group $S(4)$}
In this case we have three nontivial irrpes labelled by partitions  $\lambda_1=(3,1)$, $\lambda_2=(2,2)$, $\lambda_3=(2,1,1)$ and two one-dimensional denoted by $\id$ and $\sgn$. Everything in this section is computed in the Young-Yamanouchi representation~\cite{JJF}. Here we construct a trace-preserving ICLM and in particular ICQC for the irrep $U$ characterised by the  partition $
\lambda_1$.
The generators in the Young-Yamanouchi representation for the partition $\lambda_1$ have the following form:
\be
\varphi^{\lambda_1}(12)=\begin{pmatrix} 1 & 0 & 0\\ 0 & 1 & 0\\ 0 & 0  & -1 \end{pmatrix},\quad \varphi^{\lambda_1}(23)=\begin{pmatrix} 1 & 0 & 0\\ 0 & -\frac{1}{2} & \frac{\sqrt{3}}{2}\\ 0 & \frac{\sqrt{3}}{2} & \frac{1}{2} \end{pmatrix},\quad \varphi^{\lambda_1}(34)=\begin{pmatrix}-\frac{1}{3} & \frac{\sqrt{8}}{3} & 0\\ \frac{\sqrt{8}}{3} & \frac{1}{3} & 0\\ 0 & 0 & 1  \end{pmatrix}.
\ee
Decomposition given by~\cref{eq15} in this case reads as:
\begin{equation}
U\otimes U^c=U^{\id}\oplus U^{\lambda_1}\oplus U^{\lambda_2}\oplus U^{\lambda_3},\quad \text{and}\quad \dim\left[\Int_{S(4)}\left(U \otimes U^c \right)  \right]=4. 
\end{equation}
We see that $\sgn$ irrep does not occur in the decomposition of $U\ot U^c$. In this case we have $\Theta=\left\lbrace \id, \lambda_1, \lambda_2,\lambda_3\right\rbrace $ and $\Phi^{\lambda_1}=\Pi^{\id}+l_{\lambda_1}\Pi^{\lambda_1}+l_{\lambda_2}\Pi^{\lambda_2}+l_{\lambda_3}\Pi^{\lambda_3}$ which is the trace-preserving  ICLM (see~\Cref{prop12}).  The matrix representation of the trace-preserving ICLM $\widetilde{\Phi}^{\lambda_1}$ in the Young-Yamanouchi representation is given by:
\begin{equation}
\widetilde{\Phi}^{\lambda_1}=\begin{pmatrix}
a_1 & 0 & 0 & 0 & a_2 & 0 & 0 & 0 & a_2\\

0 & a_3 & 0  & a_4 & a_5 & 0 & 0 & 0 & -a_5\\

0 & 0 & a_3 & 0 & 0 & -a_5 & a_4 & -a_5 & 0\\

0 & a_4 & 0 & a_3 & a_5 & 0 & 0 & 0 & -a_5\\

a_2 & a_5 & 0 & a_5 & a_6 & 0 & 0 & 0 & a_7\\

0 & 0 & -a_5 & 0 & 0 & a_8 & -a_5 & a_9 & 0\\

0 & 0 & a_4 & 0 & 0 & -a_5 & a_3 & -a_5 & 0\\

0 & 0 & -a_5 & 0 & 0 & a_9 & -a_5 & a_8 & 0\\

a_2 & -a_5 & 0 & -a_5 & a_7 & 0 & 0 & 0 & a_6
\end{pmatrix},
\end{equation}
where 
\begin{equation}
\begin{split}
a_1&=\frac{1}{3}(1+2l_{\lambda_1}), \ x_2=\frac{1}{3}(1-l_{\lambda_1}),\\
a_3&=\frac{1}{6}(l_{\lambda_1}+2l_{\lambda_2}+3l_{\lambda_3}), \  a_4=\frac{1}{6}(l_{\lambda_1}+2l_{\lambda_2}-3l_{\lambda_3}),\\
a_5&=\frac{1}{3\sqrt{2}}(l_{\lambda_2}-l_{\lambda_1}), \  a_6=\frac{1}{6}(2+3l_{\lambda_1}+l_{\lambda_2}),\\ a_7&=\frac{1}{6}(2-l_{\lambda_1}-l_{\lambda_2}), \  a_8=\frac{1}{6}(2l_{\lambda_1}+l_{\lambda_2}+3l_{\lambda_3}),\\ a_9&=\frac{1}{6}(2l_{\lambda_1}+l_{\lambda_2}-3l_{\lambda_3}),
\end{split}
\end{equation}
and $l_{\lambda_1},l_{\lambda_2},l_{\lambda_3} \in \mathbb{R}$, since all characters for the group $S(4)$ are real (see explanation below~\Cref{hem}).
The corresponding Choi-Jamio{\l}kowski image (given in ~\Cref{prop13}) can be written as:
\begin{equation}
\label{choiS4}
\begin{split}
J\left(\Phi_Y^{\lambda_1} \right)=\begin{pmatrix}
a_1 & 0 & 0 & 0 & a_3 & 0 & 0 & 0 & a_3\\
0 & a_2 & 0 & a_6 & -a_9 & 0 & 0 & 0 & a_9\\
0 & 0 & a_2 & 0 & 0 & a_9 & a_6 & a_9 & 0\\
0 & a_6 & 0 & a_2 & a_9 & 0 & 0 & 0 & a_9\\
a_3 & -a_9 & 0 & a_9 & a_4 & 0 & 0 & 0 & a_5\\
0 & 0 & a_9 & 0 & 0 & a_7 & a_9 & a_8 & 0\\
0 & 0 & a_6 & 0 & 0 & a_9 & a_2 & a_9 & 0\\
0 & 0 & a_9 & 0 & 0 & a_8 & a_9 & a_7 & 0\\
a_3 & a_9 & 0 & a_9 & a_5 & 0 & 0 & 0 & a_4 
\end{pmatrix}.
\end{split}
\end{equation}
whose eigenvalues are given by (using~\Cref{cor15}):
\begin{equation}
\begin{split}
 \epsilon_1^{\lambda_1}&=\epsilon_2^{\lambda_1}=\epsilon_3^{\lambda_1}=\frac{1}{6}\left(2+3l_{\lambda_1}-2l_{\lambda_2}-3l_{\lambda_3} \right),\\ \epsilon_1^{\lambda_2}&=\epsilon_2^{\lambda_2}=\frac{1}{6}\left(2-3l_{\lambda_1}+4l_{\lambda_2}-3l_{\lambda_3} \right),\\ \epsilon_1^{\lambda_1}&=\epsilon_2^{\lambda_3}=\epsilon_3^{\lambda_3}=\frac{1}{6}\left(2-3l_{\lambda_1}-2l_{\lambda_2}+3l_{\lambda_3} \right),\\ \epsilon^{\id}&=\frac{1}{3}\left(1+3l_{\lambda_1}+2l_{\lambda_2}+3l_{\lambda_3} \right).
\end{split}
\end{equation}
Note that $ J\left( \Phi_Y^{\lambda_1}\right)\geq 0$ if and only if the parameters $l_{\lambda_1},l_{\lambda_2},l_{\lambda_3}$ belong to the shaded region shown in~\Cref{F01}. 
\begin{figure}[h!]
	\begin{center}
		\includegraphics[width=3.5in]{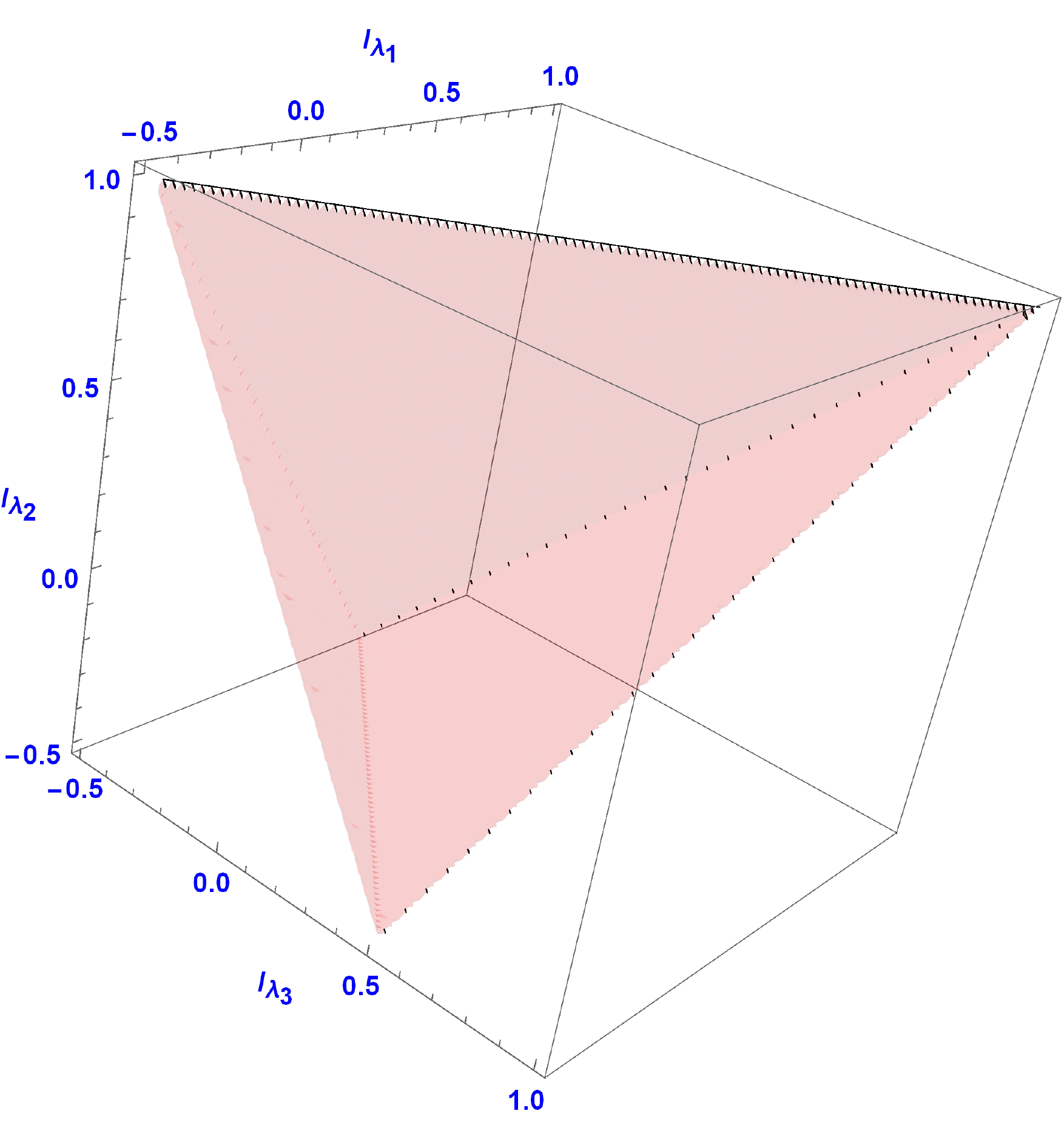}
	\end{center}
	\caption{On this figure we present allowed range of parameters $l_{\lambda_1},l_{\lambda_2},l_{\lambda_3}$ for which the trace-preserving ICLM generated by irreducible representation labelled by $\lambda_1=(3,1)$ is an ICQC. }
	\label{F01}
\end{figure}
The corresponding Kraus operators (see~\Cref{Kraus} and~\Cref{KK1}) are given as follows:
\be
\begin{array}{ll}
	K_1(\lambda_1)=\sqrt{\epsilon_1^{\lambda_1}}\begin{pmatrix}
		-\frac{2}{3} & \frac{1}{3\sqrt{2}} & 0\\
		\frac{1}{3\sqrt{2}} & 0 & 0\\
		0 & 0 & \frac{2}{3}
	\end{pmatrix},& K_2(\lambda_1)=\sqrt{\epsilon_2^{\lambda_1}}\begin{pmatrix}
	0 & 0 & \frac{1}{\sqrt{6}}\\
	0 & 0 & \frac{1}{\sqrt{3}}\\
	\frac{1}{\sqrt{6}} & \frac{1}{\sqrt{3}} & 0
\end{pmatrix},\\ K_3(\lambda_1)=\sqrt{\epsilon_3^{\lambda_1}}\begin{pmatrix}
-\frac{\sqrt{2}}{3} & -\frac{1}{3} & 0\\
-\frac{1}{3} & \frac{1}{\sqrt{2}} & 0\\
0 & 0 & -\frac{1}{3\sqrt{2}}
\end{pmatrix},&
K_1(\lambda_2)=\sqrt{\epsilon_1^{\lambda_2}}\begin{pmatrix}
	0 & -\frac{1}{\sqrt{3}} & 0\\
	-\frac{1}{\sqrt{3}} & -\frac{1}{\sqrt{6}} & 0\\
	0 & 0 & \frac{1}{\sqrt{6}}	\end{pmatrix},\\ K_2(\lambda_2)=\sqrt{\epsilon_2^{\lambda_2}}\begin{pmatrix}
	0 & 0 & -\frac{1}{\sqrt{3}}\\
	0 & 0 & \frac{1}{\sqrt{6}}\\
	-\frac{1}{\sqrt{3}} & \frac{1}{\sqrt{6}} & 0
\end{pmatrix},&
K_1(\lambda_3)=\sqrt{\epsilon_1^{\lambda_3}}\begin{pmatrix}
	0 & 0 & 0\\
	0 & 0 & -\frac{1}{\sqrt{2}}\\
	0 & \frac{1}{\sqrt{2}} & 0
\end{pmatrix},\\
K_2(\lambda_3)=\sqrt{\epsilon_2^{\lambda_3}}\begin{pmatrix}
	0 & 0 & -\frac{1}{\sqrt{2}}\\
	0 & 0 & 0\\
	\frac{1}{\sqrt{2}} & 0 & 0
\end{pmatrix},& K_3(\lambda_3)=\sqrt{\epsilon_3^{\lambda_3}}\begin{pmatrix}
0 & -\frac{1}{\sqrt{2}} & 0\\
\frac{1}{\sqrt{2}} & 0 & 0\\
0 & 0 & 0
\end{pmatrix},\\ K(\id)=\sqrt{\epsilon_1^{\id}}\begin{pmatrix}
\frac{1}{\sqrt{3}} & 0 & 0\\
0 & \frac{1}{\sqrt{3}} & 0\\
0 & 0 & \frac{1}{\sqrt{3}}
\end{pmatrix}.
\end{array}
\ee
One can check, that the resulting channel is unital and trace-preserving.

\subsection{Families of entanglement breaking ICQCs}
\label{EB}
In this section we present two families of entanglement breaking ICQCs based on $S(3)$ and the quaternion group $Q$. 
Let
\be
\label{Qset}
\mathcal{S}(\mathcal{H}):=\{\rho \in \mathcal{B}(\mathcal{H}) \ | \ \rho \geq 0, \tr \rho=1\},
\ee
denote the set of all states (density matrices) acting on the Hilbert space 
$\mathcal{H}\simeq \mathbb{C}^d$. Here $\mathcal{B}(\mathcal{H})$ denotes the algebra of all linear operators acting on $\mathcal{H}$. 
Suppose now that we are dealing with two finite dimensional Hilbert spaces $\mathcal{H}_A,\mathcal{H}_B$. 

A bipartite state $\rho_{AB} \in \mathcal{S}(\mathcal{H}_A\ot \mathcal{H}_B)$ belongs to the set of separable states $\mathcal{SEP}$ if it can be written as $\rho_{AB}=\sum_i p_i \rho_i^A \ot \sigma_i^B$, where $\rho_i^A,\sigma_i^B$ are states on $\mathcal{H}_A$ and $\mathcal{H}_B$ respectively, and $p_i$ are some positive numbers satisfying $\sum_i p_i=1$. Otherwise the state $\rho_{AB}$ is entangled.
A given quantum channel $\Phi$ (not necessarily an ICQC) is entanglement breaking (EB)~\cite{EBB} if and only if its Choi-Jamio{\l}kowski image given by
\be
\label{Choisep}
\left( \text{\noindent
	\(\mathds{1}\)}\ot \Phi \right)\left( |\psi^+\>\<\psi^+|\right) \equiv J\left(\Phi \right)  
\ee
is separable for $|\psi^+\>=\frac{1}{\sqrt{d}}\sum_{i}|ii\>$. In general we do not have a unique criterion for checking separability, but it is known that in the case of quantum states on $\mathcal{S}(\mathbb{C}^2\ot\mathbb{C}^2)$ and $\mathcal{S}(\mathbb{C}^2\ot \mathbb{C}^3)$ necessary and sufficient conditions for separability are given in terms of partial transposition $\text{\noindent
	\(\mathds{1}\)}_A\ot T_B$, where $T_B$ denotes standard transposition on $\mathcal{H}_B$. Namely, we have that $\sigma_{AB}$ is separable if and only if it has a positive partial transpose (PPT), i.e.~$(\text{\noindent
	\(\mathds{1}\)}_A\ot T_B)\sigma_{AB} \geq 0$~\cite{SepHHH,SepAP}. Since the maximal possible dimension of irreps of both $S(3)$ and $Q$ is two, we can directly apply the above mentioned criterion
to deduce when quantum channels, which are irreducibly covariant with respect to them, are also EB.

For the group $S(3)$ in the $\epsilon-$representation and partition $\lambda=(2,1)$  we have
\be
J\left(\Phi_{\epsilon} \right)\in \mathcal{SEP} \ \Leftrightarrow \ \left(\text{\noindent
	\(\mathds{1}\)}\ot T \right) J\left(\Phi_{\epsilon} \right)\geq 0 \ \Leftrightarrow \ -1\leq l_{\sgn}\leq 1, \ \frac{1}{2}(1-l_{\sgn})\geq |l_{\lambda}|.
\ee
Comparing the above conditions for PPT with the conditions for CPTP of the map $\Phi_{\epsilon}$ given in~\cref{ecp} we see that
\be
\Phi_{\epsilon} \ \text{is EB} \ \Leftrightarrow \ \begin{cases}
	-1\leq l_{\sgn}\leq 0, \ |l_{\lambda}|\leq \frac{1}{2}(1+l_{\sgn}), \\
	0<l_{\sgn}\leq 1, \  \  \ |l_{\lambda}|\leq \frac{1}{2}(1-l_{\sgn}).
\end{cases}
\ee
The solution of the above inequalities is graphically represented on the right panel of~\Cref{S3a}.

For the quaternion group $Q$ and irrep $t_4$ given in~\Cref{Quat} the separability criterion of the Choi-Jamio{\l}kowski image reads
\be
\label{116a}
J\left(\Phi^{t_4} \right)\in \mathcal{SEP} \ \Leftrightarrow \ \left(\text{\noindent
	\(\mathds{1}\)}\ot T \right) J\left(\Phi^{t_4} \right)\geq 0 \ \Leftrightarrow \ \begin{cases}
	\frac{1}{2}(1-l_{t_1}-l_{t_2}-l_{t_3})\geq 0,\\
	\frac{1}{2}(1+l_{t_1}+l_{t_2}-l_{t_3})\geq 0,\\
	\frac{1}{2}(1+l_{t_1}-l_{t_2}+l_{t_3})\geq 0,\\
	\frac{1}{2}(1-l_{t_1}+l_{t_2}+l_{t_3})\geq 0.
\end{cases}
\ee
Comparing the set of solutions of the above inequalities with the range of parameters $l_{t_1},l_{t_2},l_{t_3}$ for which the map $\Phi^{t_4}$ is CPTP, we get allowed triples $(l_{t_1},l_{t_2},l_{t_3})$ when $\Phi^{t_4}$ is an ICQC and EB channel. We present a graphical representation of allowed triples of parameters in~\Cref{S31}.
\begin{figure}[h]
\begin{center}
	\subfloat[]{\includegraphics[width=0.48\textwidth]{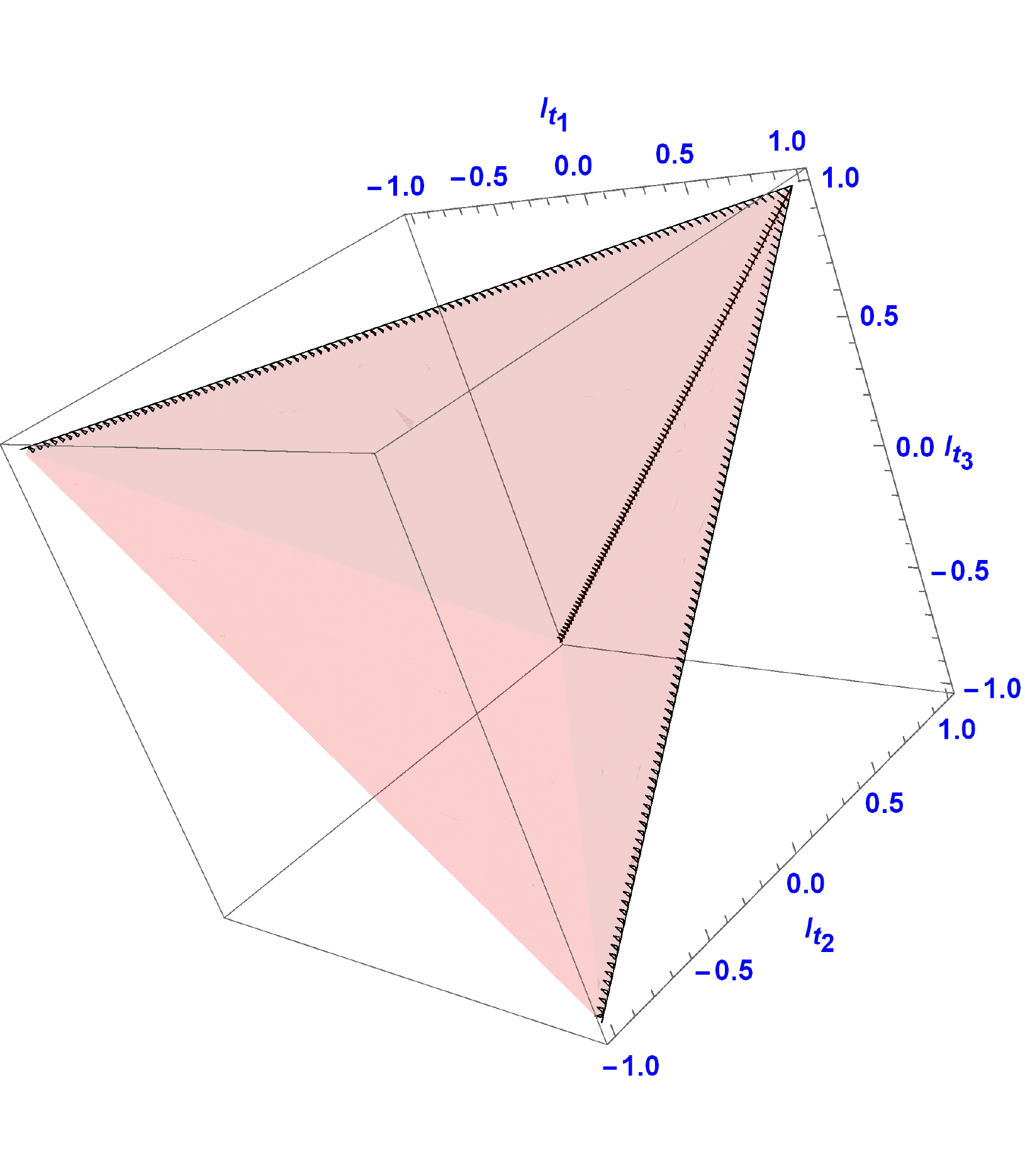}}
	\hfill
	\subfloat[]{\includegraphics[width=0.48\textwidth]{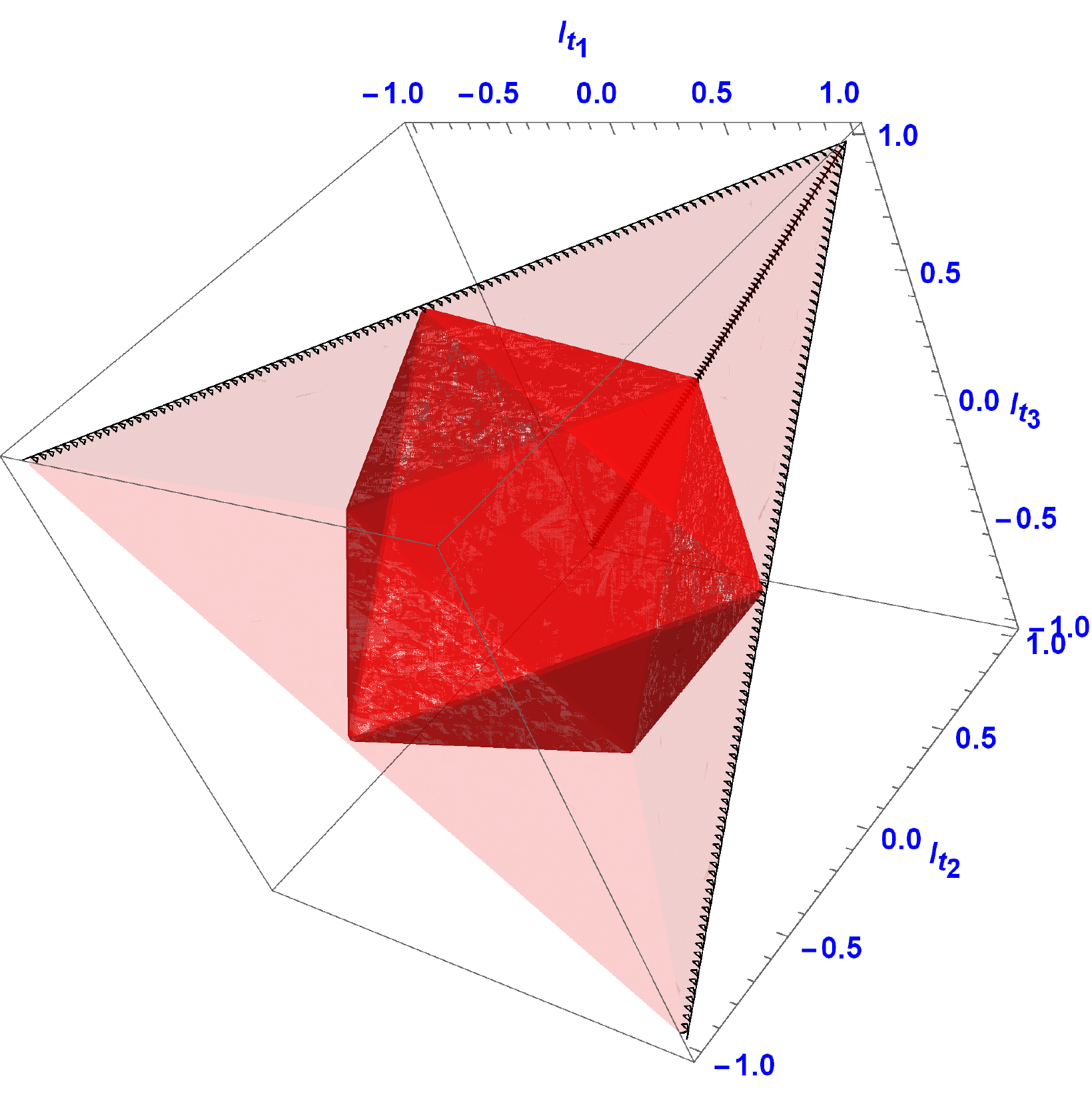}}
	\caption{On the left panel we present allowed range of parameters $l_{t_1},l_{t_2},l_{t_3}$ for which the inequalities~\cref{cc1} are satisfied and a trace-preserving ICLM generated by the 
irrep $t_4$ is an ICQC. On the right panel the grey region superimposed on the CPTP region corresponds to the range of parameters $l_{t_1},l_{t_2},l_{t_3}$ for which the ICQC $\Phi_{t_4}$ is also EB.}
	\label{S31}
\end{center}
\end{figure}

\section{Additional results: spectral properties of the rank-one projectors $\Pi_i^\alpha$}
\label{S7}
In this section, we state properties of the eigenvectors $V_{i}^{\alpha }(s,t)$
of the rank-one projectors $\Pi_{i}^{\alpha }$, defined in Proposition \ref{propP}.
To formulate the main result of this section, which is stated in~\Cref{thm21}, we need the following lemma (which is proved in \Cref{appC}).

\begin{lemma}
	\label{l27}
	Suppose that $\gamma \notin \Theta $, where $U\otimes U^{c}=
	\bigoplus_{\alpha \in \Theta }\varphi ^{\alpha }$, then 
	\begin{equation}
	\forall s,t=1,\ldots ,n\quad  V_{i}^{\gamma }(s,t) =\frac{|\varphi^{\gamma} |}{|G|}\sum_{g\in G}\varphi _{ii}^{\gamma }\left(
	g^{-1}\right) U_{C(s)}(g)U_{R(t)}\left( g^{-1}\right) =0.
	\end{equation}
\end{lemma}

\begin{theorem}
\label{thm21}
The eigenvectors $V_{i}^{\alpha } \equiv V_{i}^{\alpha }(s,t) \in  \M(n, \cC)$
of the rank-one projectors $\Pi_{i}^{\alpha } \in \End\left[ \M(n, \cC)\right]$, defined in Proposition \ref{propP}, satisfy the following properties:
\begin{enumerate}
\item 
{\begin{equation}
	\label{thm-30-1}
||V_{i}^{\alpha }(s,t)||_2^{2}=(\tPi_{i}^{\alpha })_{st,st}, \quad {\hbox{and hence}}\quad \sum_{s,t=1}^n||V_{i}^{\alpha }(s,t)||_2^{2}=1.
\end{equation}}
\item{The following summation rules hold:
\begin{align}\label{sum1}
\forall s,t=1,\ldots ,n\quad \sum_{\alpha \in \Theta }\sum_{i=1}^
{|\varphi^{\alpha }|}V_{i}^{\alpha }(s,t)=E_{st},
\end{align}
where $\{E_{st}\}_{s,t=1}^n$ is a natural basis of $\M(n,\mathbb{C})$, and 
\begin{equation}
\label{suum2}
\forall \alpha \in \Theta,\,\,i=1,\ldots ,|\varphi^{\alpha} | \qquad \sum_{s=1,\ldots
,n}V_{i}^{\alpha }(s,s)=\delta ^{\alpha ,\id}\text{\noindent
\(\mathds{1}\)}_{n},
\end{equation}
where $\text{\noindent
	\(\mathds{1}\)}_{n}$ denotes the identity matrix in $\M(n, \cC)$.}
\item{Moreover,
\begin{align}
\label{thm-30-3}
\tr V_{i}^{\alpha }(s,t) &=\delta_{\alpha, \id}\delta_{st}.
\end{align}}
\end{enumerate}
\end{theorem}

\section{Conclusions and Open Questions}
In this paper we present a detailed characterization of linear maps which are covariant with respect to an irreducible representation $U$ of a finite group $G$, in the case in which $U\ot U^c$ is simply reducible (or multiplicity-free); here $U^c$ denotes the contragradient representation. We refer to such a linear map as an irreducibly covariant linear map (ICLM). We derive necessary and sufficient conditions under which an ICLM is trace-preserving and completely positive, and is thus an irreducibly covariant quantum channel (ICQC). These conditions are obtained by requiring that the Choi-Jamio{\l}kowski image of the ICLM is positive semidefinite. We present explicit analytical expressions for the eigenvalues and the eigenvectors of the Choi-Jamio{\l}kowski image, and give the conditions for complete positivity as a set of inequalities, which can be solved directly for any given finite group $G$. The resulting characterization of the ICQC is given entirely in terms of representation characteristic of the group $G$, and is valid for any finite group and its irreps $U$, as long as $U\ot U^c$ is simply reducible. 
Moreover, we construct a wide class of non-trivial ICQCs for which the conditions for complete positivity are automatically satisfied. In addition, we obtain a geometrical interpretation of the spectrum of the Choi-Jamio{\l}kowski image of an ICQC, by showing that it always lies in the intersection of a certain scaled simplex and a linear subspace, which are defined by the spectral properties of the projectors appearing in the spectral decomposition of an ICQC.

As a direct application of our result, we give the full description of ICQCs generated by irreps ($U$) of certain finite groups, for which $U\ot U^c$ is simply reducible. These are the quaternion group $Q$, and the symmetric groups $S(3)$ and $S(4)$. In each case we present both the matrix representation and the Kraus representation of the ICQC. Moreover, for covariance with respect to two-dimensional irreps of the groups $S(3)$ and $Q$, we obtain conditions under which the ICQCs are also entanglement breaking channels.  

We expect to obtain an analogous characterization of quantum channels which are irreducibly covariant with respect to a compact group.  However, such an analysis needs to be done carefully, since dealing with compact groups is technically more challenging than finite groups. This is the topic worthy of a separate project. In our opinion, however, it might be more interesting to first extend   our analysis to the cases in which $(i)$ $U \otimes U^c$ is {\em{not}} simple reducible, and $(ii)$ the linear map is positive but not completely positive. In particular, the latter might lead to new classes of (covariant) entanglement witnesses.
	
\section*{Acknowledgments}
M.S.~is supported by the grant "Mobilno{\'s}{\'c} Plus IV", 1271/MOB/IV/2015/0 from the Polish Ministry of Science and Higher Education. M.M.~would like to thank Pembroke College and the Centre for Mathematical Sciences at the University of Cambridge for hospitality. N.D.~is grateful to A.~Holevo and R.~Werner for helpful discussions.
\appendix
\section{Proofs of results from~\Cref{S4}}
\label{appA}
\begin{proof}[Proof of~\Cref{def7 copy(1)}]
	Suppose that $\mat(\Phi )\in \Int_{G}\left( U\otimes \overline{U}%
\right),$ then $\forall g \in G$ and $\forall X \in \M(n,\mathbb{C})$ by~\Cref{L2} we have the following chain of equivalences:
\begin{equation}
\begin{split}
& \left( \sum_{\nu=1}^{m}A^{\nu}\otimes \overline{B}
^{\nu}\right) U(g)\otimes \overline{U}(g)=U(g)\otimes \overline{U}%
(g)\left( \sum_{\nu=1}^{m}A^{\nu}\otimes \overline{B}^{\nu}\right) \Leftrightarrow \\
&
\left( \sum_{\nu=1}^{m}A^{\nu}\otimes \overline{B}^{\nu}\right) U(g)\otimes \overline{U}%
(g).\cV(X)=U(g)\otimes \overline{U}(g)\left( \sum_{\nu=1}^{m}A^{\nu}\otimes \overline{B}^{\nu}\right) .\cV(X)\Leftrightarrow \\
&
\left( \sum_{\nu=1}^{m}A^{\nu}\otimes \overline{B}^{\nu}\right) \cV\left( U(g)X U^{\dagger}(g)\right) =U(g)\otimes 
\overline{U}(g)\cV\left( \sum_{\nu=1}^{m}A^{\nu}X(B^{\nu})^{\dagger}\right)\Leftrightarrow \\
&
\cV\left( \sum_{\nu=1}^{m}A^{\nu}\Ad_{U(g)}(X)(B^{\nu})^{\dagger}\right) =\cV\left( U(g)
\sum_{\nu=1}^{m}(A^{\nu}X(B^{\nu})^{\dagger})U^{\dagger}(g)\right) \Leftrightarrow \\
& \cV\left( \Phi 
\left[ \Ad_{U(g)}(X)\right]\right) =\cV\left( \Ad_{U(g)}\left[ \sum_{\nu=1}^{m} A^{\nu}X\left(B^{\nu}\right)^{\dagger}\right] \right) .
\end{split}
\end{equation}
\end{proof}

\section{Proofs of results from~\Cref{S5}}
\label{appB}

\begin{proof}[Proof of~\Cref{prop8}]
From the decomposition given by~\cref{eq15} and Schur's Lemma it is clear, that dimension of the
commutant space of the representation $U\otimes U^{c}$ is equal to number of 
irreps $\varphi^{\alpha}$ appearing in this decomposition. From the orthogonality of the projectors $\widetilde{\Pi}^{\alpha}$ it follows that they are linearly independent, so they form a basis of the
space $\Int_{G}\left( U\otimes U^{c}\right)$. The
properties of the operators $\widetilde{\Pi}^{\alpha}$ are derived in~\cite{NS}, where in
particular is shown that, for a given $\varphi^{\alpha}$, the operator $\widetilde{\Pi}^{\alpha}$ (\cref{eq17}) is an orthogonal projector onto a subspace of $\mathbb{C}^{n^{2}}$ containing all irreps $\varphi ^{\alpha }$ included in $U\otimes U^{c}$. Because, by assumption, in our case each irrep $\varphi^{\alpha}$ appears at most once in $U\otimes U^{c}$, then in our case projectors $\widetilde{\Pi}^{\alpha}$ project onto the subspace of the irrep $\varphi^{\alpha}$ in the representation space $\mathbb{C}^{n^{2}}$. The orthogonality of the projectors $\widetilde{\Pi}^{\alpha}$ follows now
from the direct sum decomposition in~\cref{eq15}. It can be also checked by direct calculation, using the orthogonality relations for irreducible characters~\cref{or}.

Now we prove that $\widetilde{\Pi}^{\alpha}\in \Int_{G}\left( U\otimes U^{c}\right)$ for $\alpha \in
\Theta$. In fact we have for any $h\in G$:
\be
\begin{split}
&U(h)\otimes \overline{U}(h)\widetilde{\Pi}^{\alpha}U\left( h^{-1}\right) \otimes \overline{U}\left( h^{-1}\right) =
\frac{|\varphi^{\alpha}|}{|G|}\sum_{g\in G}\chi ^{\alpha
}\left( g^{-1}\right) U\left( hgh^{-1}\right) \otimes \overline{U}\left( hgh^{-1}\right)\\
&\frac{|\varphi^{\alpha}|}{|G|}\sum_{k\in G}\chi ^{\alpha
}\left( h^{-1}k^{-1}h\right) U(k)\otimes \overline{U}(k)=\frac{|\varphi^{\alpha}|}{|G|}%
\sum_{k\in G}\chi ^{\alpha }\left( k^{-1}\right) U(k)\otimes \overline{U}(k)=\widetilde{\Pi}^{\alpha}, 
\end{split}
\ee
because characters are class functions.  
\end{proof}

\begin{proof}[Proof of~\Cref{prop11}]
	The hermiticity of matrices $\widetilde{\Pi}_i^{\alpha}$ follows directly from the unitarity of the representation $U:G\rightarrow \M(n,\mathbb{C})$. We have also:
	\be
	\begin{split}
	\widetilde{\Pi}_i^{\alpha}\widetilde{\Pi}_j^{\alpha}&=\frac{ |\varphi^{\alpha}|^{2}}{|G|^{2}}%
	\sum_{g,h\in G}\varphi _{ii}^{\alpha }\left( g^{-1}\right) \varphi _{jj}^{\alpha
	}\left( h^{-1}\right) U(gh)\otimes \overline{U}(gh) \\
	&=\frac{|\varphi^{\alpha}|^{2}}{|G|^{2}}\sum_{g,k\in G}\varphi
	_{ii}^{\alpha }\left( g^{-1}\right) \varphi _{jj}^{\alpha }\left( k^{-1}g\right) U(k)\otimes \overline{%
		U}(k) \\
	&=\frac{|\varphi^{\alpha}| ^{2}}{|G|^{2}}\sum_{g,k\in G}\sum_{l}\varphi
	_{ii}^{\alpha }\left( g^{-1}\right) \varphi _{jl}^{\alpha }\left( k^{-1}\right) \varphi _{lj}^{\alpha
	}(g)U(k)\otimes \overline{U}(k).
	\end{split}
	\ee
	Now using the Schur orthogonality relations for irreps, given by \cref{or}, we get:
	\be
	\begin{split}
	\widetilde{\Pi}_i^{\alpha}\widetilde{\Pi}_j^{\alpha}&=\frac{|\varphi^{\alpha}|}{|G|}\sum_{k\in
		G}\sum_{l}\delta _{il}\delta _{ij}\varphi _{jl}^{\alpha }\left( k^{-1}\right) U(k)\otimes 
	\overline{U}(k)\\
	&=\frac{|\varphi^{\alpha}|}{|G|}\sum_{k\in G}\delta _{ij}\varphi
	_{ji}^{\alpha }\left( k^{-1}\right) U(k)\otimes \overline{U}(k)=\delta _{ij}\widetilde{\Pi}_i^{\alpha}. 
	\end{split}
	\ee
	Similarly using the orthogonality relation~\cref{or} for irreps, and the
	decomposition in~\cref{eq15} one can prove that $\tr ( \widetilde{\Pi}_i^{\alpha}) =1$. 
\end{proof}

\begin{proof}[Proof of~\Cref{cor17}]
	The first property in~\cref{chain} is obvoius from the explicit form of the projectors $\Pi_i^{\alpha}$ given in~\cref{Pii_a}.
	
	Now we prove directly the second property form~\cref{chain}. Namely $\forall X\in \M(n,\mathbb{C})$ we have:
	\be
	\Pi _{i}^{\alpha }\Pi _{j}^{\beta }(X)=\frac{|\varphi^{\alpha}|^{2}}{|G|^{2}}\sum_{g,h\in G}\varphi _{ii}^{\alpha }\left( g^{-1}\right) \varphi
	_{jj}^{\beta }\left( h^{-1}\right) U(gh)XU\left( h^{-1}g^{-1}\right) .
	\ee
Setting $s=gh$ we get the following:
	\be
	\begin{split}
	\Pi _{i}^{\alpha }\Pi _{j}^{\beta }(X)&=\frac{|\varphi^{\alpha}|^{2}}{|G|^{2}}\sum_{g,s\in G}\varphi
	_{ii}^{\alpha }\left( g^{-1}\right) \varphi _{jj}^{\beta }\left( s^{-1}g\right) U(s)XU\left( s^{-1}\right) \\
	&=\frac{|\varphi^{\alpha}|^{2}}{|G|^{2}}\sum_{g,s\in G}\sum_{l}\varphi
	_{ii}^{\alpha }\left( g^{-1}\right) \varphi _{jl}^{\beta }\left( s^{-1}\right) \varphi _{lj}^{\beta
	}(g)U(s)XU\left( s^{-1}\right) .
	\end{split}
	\ee
	Using the orthogonality relation~\cref{or} for irreps, we get for any $X\in \M(n,\mathbb{C})$:
	\be
	\begin{split}
	\Pi _{i}^{\alpha }\Pi _{j}^{\beta }(X)&=\frac{|\varphi^{\alpha}|}{|G|}%
	\sum_{s\in G}\sum_{l}\delta ^{\alpha \beta }\delta _{il}\delta _{ij}\varphi
	_{jl}^{\alpha }\left( s^{-1}\right) U(s)XU\left( s^{-1}\right) \\
	&=\frac{|\varphi^{\alpha}|}{|G|}\sum_{s\in G}\delta ^{\alpha \beta
	}\delta _{ij}\varphi _{ji}^{\alpha }\left( s^{-1}\right) U(s)XU\left( s^{-1}\right) =\delta ^{\alpha
	\beta }\delta _{ij}\Pi _{i}^{\alpha }(X).
\end{split}
\ee
To prove the third property from~\cref{chain}, which is in fact, a
conjugation of operators in the space $\End[\M(n,\mathbb{C})]$ we use general considerations from~\Cref{nots}. We have $\forall X,Y\in \M(n,\mathbb{C})$:
\be
\begin{split}
(X,\Pi _{i}^{\alpha }(Y))&=\tr\left( X^{\dagger}\frac{|\varphi^{\alpha}|}{|G|}
\sum_{g\in G}\varphi _{ii}^{\alpha }\left( g^{-1}\right) U(g)YU\left( g^{-1}\right) \right) \\
&=\tr\left( \frac{|\varphi^{\alpha}|}{|G|}\sum_{g\in G}\varphi _{ii}^{\alpha
}\left( g^{-1}\right) U\left( g^{-1}\right) X^{\dagger}U(g)Y\right) \\
&=\tr\left( \left[ \frac{|\varphi^{\alpha}|}{|G|}
\sum_{g\in G}\overline{\varphi }_{ii}^{\alpha }\left( g^{-1}\right) U\left( g^{-1}\right) XU(g)\right]^{\dagger}Y\right) \\
&=\tr\left( \left[ \frac{|\varphi^{\alpha}|}{|G|}\sum_{g\in G}\varphi
_{ii}^{\alpha }(g)U\left( g^{-1}\right) XU(g)\right]^{\dagger}Y\right) \\
&=\tr\left( \left[ \frac{|\varphi^{\alpha}|}{
	|G|}\sum_{h\in G}\varphi _{ii}^{\alpha
}\left( h^{-1}\right) U(h)XU\left( h^{-1}\right) \right]^{\dagger}Y\right) =(\Pi _{i}^{\alpha }(X),Y).
\end{split}
\ee
\end{proof}

\begin{proof}[Proof of~\Cref{propP}]
	First we prove that the matrices  $V_{r}^{\alpha}(s,t)$ satisfy the
	eigenvalue equation for the projector ${\Pi}_r^{\alpha}$. For
	any $s,t=1,\ldots,n$ we can write
	\be
	\begin{split}
		\Pi_r^{\alpha}\left( V_{r}^{\alpha }(s,t)\right) &=\frac{|\varphi^{\alpha}|^{2}}{
		|G|^{2}} \sum_{g,h\in G}\varphi _{rr}^{\alpha }\left( h^{-1}\right) \varphi
	_{rr}^{\alpha }\left( g^{-1}\right) U(h)U_{c(s)}(g)U_{R(t)}\left( g^{-1}\right) U\left( h^{-1}\right)\\ 
	&=\frac{|\varphi^{\alpha}|^{2}}{|G|^{2}} \sum_{g,h\in G}\varphi
	_{rr}^{\alpha }\left( h^{-1}\right) \varphi _{rr}^{\alpha
	}\left( g^{-1}\right) U_{C(s)}(hg)U_{R(t)}\left( g^{-1}h^{-1}\right)\\ 
	&=\frac{|\varphi^{\alpha}|^{2}}{|G|^{2}} \sum_{w,h\in G}\varphi
	_{rr}^{\alpha }\left( h^{-1}\right) \varphi _{rr}^{\alpha
	}\left(hw^{-1}\right) U_{C(s)}(w)U_{R(t)}\left( w^{-1}\right)\\
	&=\frac{|\varphi^{\alpha}|}{|G|} \sum_{w\in G}\varphi _{rr}^{\alpha
	}\left( w^{-1}\right) U_{C(s)}(w)U_{R(t)}\left( w^{-1}\right)  =V_{r}^{\alpha }(s,t), 
	\end{split}
	\ee
	where in the last step we have used the orthogonality relation~\cref{or} for
	irreducible representations. Now it remains to prove that not all matrices $
	V_{r}^{\alpha }(s,t)$ are equal to zero. This can be done by calculating the
	norm of the matrices $V_{r}^{\alpha }(s,t)$.
	\be
	\label{cos1}
	||V_{r}^{\alpha }(s,t)||_2^{2}=\frac{|\varphi^{\alpha}|}{|G|}\sum_{g\in
		G}\varphi _{rr}^{\alpha }\left( g^{-1}\right) U_{ss}(g)\overline{U}_{tt}(g). 
	\ee
	Comparing the right hand side of~\cref{cos1} with the corresponding matrix element of $\widetilde{\Pi}_r^{\alpha}$,
	\be
	\label{useful}
	\left( \widetilde{\Pi}_r^{\alpha}\right)_{st,st}=\frac{|\varphi^{\alpha}|}{|G|}\sum_{g\in
		G}\varphi _{rr}^{\alpha }\left( g^{-1}\right) \left[ U(g)\otimes \overline{U}(g)\right] _{st,st}=\frac{%
		|\varphi^{\alpha}|}{|G|}\sum_{g\in G}\varphi _{rr}^{\alpha
	}\left( g^{-1}\right) U(g)_{ss}\overline{U}(g)_{tt}, 
	\ee
	we see that the norms $||V_{r}^{\alpha }(s,t)||_2^{2}$ are equal to the
	diagonal terms of the matrix $\widetilde{\Pi}_r^{\alpha}\in \M\left( n^{2},\mathbb{C}\right)$. As $\widetilde{\Pi}_r^{\alpha}$
is a projector of rank one, its diagonal terms are
	non-negative and at least one of them is positive. Otherwise the
	projector $\widetilde{\Pi}_r^{\alpha}$ would be zero, which is impossible since $\alpha\in \Theta$.
	
	In order to prove the orthonormality relation given in~\cref{VHS} it is enough to use the explicit form of the eigenvectors given in~\cref{11} for those pairs of indices for which eigenvectors are nonzero. Thanks to this we have
	\be
	\label{p1}
	\begin{split}
		&\left(V_i^{\alpha}(s,t),V_j^{\beta}(p,q) \right)=\tr\left[\left( V_i^{\alpha}(s,t)\right) ^{\dagger}V_j^{\beta}(p,q) \right]=\\
		&=\frac{1}{\sqrt{\left(\widetilde{\Pi}_i^{\alpha} \right)_{st,st}}} \frac{1}{\sqrt{\left(\widetilde{\Pi}_j^{\beta} \right)_{pq,pq}}}\frac{|\varphi^{\alpha}||\varphi^{\beta}|}{|G|^2}\sum_{g,h}\varphi_{ii}^{\alpha}(g)\varphi_{jj}^{\beta}\left(h^{-1} \right)u_{qt}\left(gh^{-1}\right)u_{sp}\left(g^{-1}h\right).   
	\end{split}
	\ee
	Using the substitution $gh^{-1}=w^{-1}$, the orthogonality relation~\cref{or} for irreps, and the expression for the matrix element of the projector $\widetilde{\Pi}_i^{\alpha}$ from~\cref{useful},
the right hand side of~\cref{p1} can be expressed as follows:
	\be
	\label{calc}
	\begin{split}
		&\frac{1}{\sqrt{\left(\widetilde{\Pi}_i^{\alpha} \right)_{st,st}}} \frac{1}{\sqrt{\left(\widetilde{\Pi}_j^{\beta} \right)_{pq,pq}}}\frac{|\varphi^{\alpha}||\varphi^{\beta}|}{|G|^2}\sum_{g,h}\varphi_{ii}^{\alpha}(g)\varphi_{jj}^{\beta}\left(g^{-1}w^{-1}\right)u_{qt}\left(w^{-1}\right)u_{sp}\left(w\right)=\\
		&=	\frac{1}{\sqrt{\left(\widetilde{\Pi}_i^{\alpha} \right)_{st,st}}} \frac{1}{\sqrt{\left(\widetilde{\Pi}_j^{\beta} \right)_{pq,pq}}}\frac{|\varphi^{\alpha}||\varphi^{\beta}|}{|G|^2}\sum_{w}\sum_{r}\left(\sum_g \varphi^{\beta}_{jr}\left(g^{-1} \right)  \varphi^{\alpha}_{ii}(g) \right) \varphi_{rj}^{\beta}\left(w^{-1} \right)\bar{u}_{tq}(w)u_{sp}(w)=\\
		&=\frac{1}{\sqrt{\left(\widetilde{\Pi}_i^{\alpha} \right)_{st,st}}} \frac{1}{\sqrt{\left(\widetilde{\Pi}_i^{\alpha} \right)_{pq,pq}}}\frac{|\varphi^{\alpha}|}{|G|}\sum_w \varphi_{ii}^{\alpha}\left(w^{-1} \right)u_{sp}(w)\bar{u}_{tq}(w)=\frac{\left(\widetilde{\Pi}_i^{\alpha} \right)_{st,pq} }{\sqrt{\left(\widetilde{\Pi}_i^{\alpha} \right)_{st,st}}\sqrt{\left(\widetilde{\Pi}_i^{\alpha} \right)_{pq,pq}}}.
	\end{split}
	\ee
	From the~\cite{Horn} we know that the result of the calculations in~\cref{calc} has modulus equal to one, so it can be expressed as $e^{\operatorname{i}\zeta}$, where the parameter $\zeta \equiv \zeta(s,t,p,q)$ for some $s,t,p,q \in \{1,\ldots,n\}$. This proves the orthonormality relation~\cref{VHS}. 
\end{proof}	

\begin{proof}[Proof of~\Cref{RpropP}]	
	From the argumentation presented in the above proof it is  clear that when $(s,t)=(p,q)$, then $\zeta=0$. Moreover, for a fixed $\alpha \in \Theta$ and $i=1,\ldots,|\varphi^{\alpha}|$, the vectors $V_i^{\alpha}(s,t)$ and $V_i^{\alpha}(p,q)$ can differ only by a phase $e^{\operatorname{i}\zeta}$, i.e.~they lie in the same ray.
\end{proof}

\begin{proof}[Proof of~\Cref{prop12}]
	We have for any $X\in \M(n,\mathbb{C})$:
	\be
	\tr[\Phi (X)]=\tr\left( \sum_{g\in G}\sum_{\alpha \in \Theta }\frac{|\varphi^{\alpha}|}{|G|}l_{\alpha }\chi ^{\alpha }\left( g^{-1}\right) U(g)XU^{\dagger}(g)\right) = 
	\ee
	\be
	=\sum_{\alpha \in \Theta }l_{\alpha }\left( \sum_{g\in G}\frac{|\varphi^{\alpha}|}{
		|G|}\chi ^{\alpha }\left( g^{-1}\right) \right) \tr(X)=l_{\id}\tr(X), 
	\ee
	where we use the orthogonality relation for irreducible characters given by~\cref{orr2}. 

	From the above, it is clear that the ICLM $\Phi \in \End\left[\M(n,\mathbb{C}) \right] $ is trace-preserving if and only if $l_{\id}=1$.
\end{proof}

\begin{proof}[Proof of~\Cref{prop13}]
	We have:
	\be
	J\left( \Pi^{\alpha }\right) =\sum_{ij}E_{ij}\otimes \Pi^{\alpha
	}(E_{ij})=\sum_{ij}E_{ij}\otimes \frac{|\varphi^{\alpha}|}{|G|}%
	\sum_{g\in G}\chi ^{\alpha }\left( g^{-1}\right) U_{C(i)}(g)U_{R(j)}^{\dagger}(g), 
	\ee
	\bigskip where 
	\be
	U_{R(i)}^{\dagger}(g)=U_{R(i)}(g^{-1})=[U_{C(i)}(g)]^{\dagger}. 
	\ee
	In particular for $\Pi^{\id}$ we have:
	\be
	J( \Pi^{\id}) =\sum_{ij}E_{ij}\otimes \frac{1}{|G|}\sum_{g\in
		G}U_{C(i)}(g)U_{R(j)}\left( g^{-1}\right) =\sum_{ij}E_{ij}\otimes \frac{1}{|U|}\delta
	_{ij}\text{\noindent
		\(\mathds{1}\)}_{n}=\frac{1}{|U|}\text{\noindent
		\(\mathds{1}\)}_{n}\otimes \text{\noindent
		\(\mathds{1}\)}_{n}, 
	\ee
	where we have used the orthogonality relation~\cref{or} for irreps. 
\end{proof}

\begin{proof}[Proof of~\Cref{prop14}]
	We prove the Proposition checking directly that the eigenvalue equation holds. We calculate:
	\be
	J\left( \Pi^{\alpha }\right) |v_{i}^{\beta }\>=\sum_{j}|j\>\otimes \frac{|\varphi^{\alpha}| }{|G|}\sum_{g\in G}\chi ^{\alpha }\left( g^{-1}\right) \tr\left( V_{i}^{\beta
} U^{\dagger}(g)\right) U(g)|j\>.
\ee
The operator
\be
X^{\alpha }\left( V_{i}^{\beta
}\right) \equiv \frac{|\varphi^{\alpha}| }{|G|}\sum_{g\in G}\chi
^{\alpha }\left( g^{-1}\right) \tr\left( V_{i}^{\beta
} U^{\dagger}(g)\right) U(g)\in \M(n,\mathbb{C}),
\ee
which appears on $RHS$ has the following important property, which one
derives by direct calculation
\be
\Pi_{k}^{\gamma }\left[ X^{\alpha }\left(V_{i}^{\beta
}\right) \right] =\delta ^{\gamma \beta }\delta
_{ki}X^{\alpha }\left(V_{i}^{\beta
}\right) \ : \ \alpha ,\beta ,\gamma \in \Theta ,\quad
k=1,\ldots,|\varphi^{\gamma}|,\quad i=1,\ldots,|\varphi^{\beta}|,
\ee
where $\Pi_{k}^{\alpha },$ $\alpha \in \Theta ,\quad k=1,\ldots,|\varphi^{\alpha}|$ are rank one projectors described in~\Cref{prop11}, so for any $\alpha
\in \Theta$, the matrices $X^{\alpha }\left( V_{i}^{\beta
}\right) $ are eigenvectors
of the rank one projectors $\Pi_{k}^{\beta}$ and therefore these matrices must be proportional to the matrices $V_{i}^{\beta
}\in \M(n,\mathbb{C})$ which are eigenvectors of $\Pi_{k}^{\beta }$ i.e.~we have:
\be
X^{\alpha }\left( V_{i}^{\beta
}\right) =\mu_{i}(\alpha ,\beta )V_{i}^{\beta
}.
\ee
Now it remains to calculate the proportionality coefficients $\mu_{i}(\alpha ,\beta)$, which can be done using the normalization of the
eigenvectors $V_{i}^{\beta
}$ of projectors $\Pi_{k}^{\beta }$ written in
their matrix representation:
\be
\tr\left(V_{i}^{\beta
} \left( V_{i}^{\beta
}\right) ^{\dagger}\right) =1.
\ee
From this it follows:
\be
\begin{split}
\mu_{i}(\alpha ,\beta )&=\tr\left( X^{\alpha }\left( V_{i}^{\beta
}\right) \left( V_{i}^{\beta
} \right) ^{\dagger}\right)\\
&=\frac{|\varphi^{\alpha}|}{|G|}\sum_{g\in G}\chi ^{\alpha
}\left( g^{-1}\right) \tr\left( V_{i}^{\beta} U^{\dagger}(g)\right) \tr\left( U(g)\left(V_{i}^{\beta
} \right) ^{\dagger}\right)\\
&=\frac{|\varphi^{\alpha}|}{|G|}\sum_{g\in G}\chi ^{\alpha
}\left( g^{-1}\right) \left| \tr\left( V_{i}^{\beta
} U^{\dagger}(g)\right) \right| ^{2}.
\label{182}
\end{split}
\ee
Finally using~\cref{11} and~\Cref{propP} we get the last result.
\end{proof}

\begin{proof}[Proof of~\Cref{r28}]
	In order to prove~\cref{eq281} it is enough to compute directly quantity $\tr\left(V_i^{\beta}U^{\dagger}(g) \right)$ using explicit form of the eigenvectors given by~\cref{11} in~\Cref{propP}.
	
	To prove~\cref{eq282} as a starting point we use~\cref{eq281} for some fixed irrep $\varphi^\alpha$. Expanding the square of the modulus, we get
	\be
	\frac{|\varphi^{\alpha}||\varphi^{\beta}|}{|G|^3}\frac{1}{\left(\widetilde{\Pi}_i^{\beta} \right)_{st,st}}\sum_{g,h,w}\chi^{\alpha}\left(g^{-1} \right)\varphi_{ii}^{\beta}\left(h^{-1}\right) \varphi_{ii}^{\beta}\left(w\right)u_{ts}\left(h^{-1}g^{-1}h \right)u_{st}\left(w^{-1}gw\right).
	\ee
	Now substituting $r=h^{-1}g^{-1}h$ in the above, we get
	\be
	\frac{|\varphi^{\alpha}||\varphi^{\beta}|}{|G|^3}\frac{1}{\left(\widetilde{\Pi}_i^{\beta} \right)_{st,st}}\sum_{h,r}\chi^{\alpha}\left(r \right)\varphi_{ii}^{\beta}\left(h^{-1}\right) \varphi_{ii}^{\beta}\left(w\right)u_{ts}\left(r \right)u_{st}\left(w^{-1}hr^{-1}h^{-1}w\right),
	\ee
	since the character $\chi^{\alpha}$ is a class function. Next, substituting $f=w^{-1}h$ we get 
	\be
	\frac{|\varphi^{\alpha}||\varphi^{\beta}|}{|G|^3}\frac{1}{\left(\widetilde{\Pi}_i^{\beta} \right)_{st,st}}\sum_{f,h,r}\chi^{\alpha}(r)\varphi_{ii}^{\beta}\left(h^{-1}\right)\left(\sum_{a}\varphi_{ia}^{\beta}(h)
	\varphi_{ai}^{\beta}\left(f^{-1} \right)  \right)u_{ts}(r)u_{st}\left(fr^{-1}f^{-1} \right).  
	\ee
	Finally, using the orthogonality relation for irreps given in~\cref{or} we obtain
	\be
	\frac{|\varphi^{\alpha}||\varphi^{\beta}|}{|G|^3}\frac{1}{\left(\widetilde{\Pi}_i^{\beta} \right)_{st,st}}\sum_{f,r}\chi^{\alpha}(r)\varphi_{ii}^{\beta}\left(t^{-1} \right)u_{ts}(r)u_{st}\left(fr^{-1}f^{-1}\right).  
	\ee
	Changing indeces in the sum $r \rightarrow g^{-1}$ and $f \rightarrow h$ we obtain the claim.
\end{proof}

\begin{proof}[Proof of~\Cref{hem}]
	If $\Phi(X)^{\dagger}=\Phi\left( X^{\dagger}\right) $, then
	\be
	J(\Phi)^{\dagger}=\sum_{i,j}E_{ji}\otimes \Phi\left( E_{ij}\right) ^{\dagger}=\sum_{i,j}E_{ji}\otimes
	\Phi\left( E_{ji}\right) =J(\Phi). 
	\ee
	If $J(\Phi)^{\dagger}=J(\Phi)$, then
	\be
	\sum_{i,j}E_{ji}\otimes \Phi(E_{ij})^{\dagger}=\sum_{i,j}E_{ji}\otimes
	\Phi(E_{ji})=\sum_{i,j}E_{ji}\otimes \Phi\left( E_{ij}^{\dagger}\right) . 
	\ee
\end{proof}

\section{Proofs of results from~\Cref{S7}}
\label{appC}
\begin{proof}[Proof of~\Cref{l27}]
	Suppose that  the matrix $V_{i}^{\gamma }(s,t)$ is nonzero, then by~\Cref{prop14} we know that it satisfies the eigenvalue equation
	\be
	\Pi_{i}^{\gamma }\left(V_{i}^{\gamma }(s,t)\right) =V_{i}^{\gamma }(s,t),
	\ee
	and this means that the projector $\Pi_{i}^{\gamma}$ is nonzero. This  implies that the projector $\Pi^{\gamma}=\sum_{i=1}^{|\varphi^{\gamma}|}\Pi_{i}^{\gamma}$ is also nonzero. The latter operator projects onto the
	subspace of irreps $\varphi^{\gamma}$ in $U\otimes U^{c}$, which does not
	exist in this representation because $\gamma \notin \Theta$. So we get a
	contradiction assuming that \ the matrix $V_{i}^{\gamma }(s,t)$ is nonzero
	when $\gamma \notin \Theta$.
\end{proof}

\begin{proof}[Proof of~\Cref{thm21}]
	The relations in~\cref{thm-30-1} and in \cref{thm-30-3} can be verified directly using Schur's  orthogonality relations.
	The proof of the first summation rule from~\cref{sum1} is based on~\Cref{l27}.
	It is clear that for any $k,l=1,\ldots,n$ we have
	\be
	\begin{split}
	\sum_{\alpha \in \Theta }\sum_{i=1}^{|\varphi^{\alpha}|}V_{i}^{\alpha
	}(s,t)&=\sum_{\alpha \in \widehat{G}}\sum_{i=1}^{|\varphi^{\alpha}|}V_{i}^{\alpha
}(s,t)=\sum_{\alpha \in \widehat{G}}\sum_{i=1}^{|\varphi^{\alpha}|}\frac{|\varphi^{\alpha}|}{|G|}\sum_{g\in G}\varphi _{ii}^{\alpha
}\left( g^{-1}\right) U_{C(s)}(g)U_{R(t)}\left( g^{-1}\right)\\
&=\frac{1}{|G|}\sum_{g\in G}\sum_{\alpha \in \widehat{G}}\chi ^{\alpha
}(\mathrm{e})\chi ^{\alpha }\left( g^{-1}\right) U_{C(s)}(g)U_{R(t)}\left( g^{-1}\right),
\end{split}
\ee
where $\mathrm{e}$ denotes the identity element of the given group $G$. Now using the orthogonality relation for irreps given by \cref{or2} we get 
\be
\sum_{\alpha \in \Theta }\sum_{i=1}^{|\varphi^{\alpha}|}V_{i}^{\alpha
}(s,t)=U_{C(s)}(\mathrm{e})U_{R(t)}(\mathrm{e})=E_{st}.
\ee
To prove the second summation rule in~\cref{suum2} it is enough to observe that 
\be
\forall g\in G\quad U_{C(s)}(g)U_{R(t)}\left( g^{-1}\right) =\text{\noindent
	\(\mathds{1}\)}_{n}\in \M(n,\mathbb{C})
\ee
and that
\be
\frac{|\varphi^{\alpha}|}{|G|}\sum_{g\in G}\varphi _{ii}^{\alpha
}\left( g^{-1}\right) =\frac{|\varphi^{\alpha}|}{|G|}\sum_{g\in G}\varphi
^{\id}(g)\varphi _{ii}^{\alpha }\left( g^{-1}\right) =\delta ^{\alpha,\id},
\ee
which follows directly from the orthogonality relation for irreps given by~\cref{orr2}. 
\end{proof}

\section{Connection between Kraus representation and Choi-Jamio{\l}kowski image of a quantum channel}
\label{appD}

\begin{lemma}
	\label{Kraus}
	Let us assume, that we are given with a quantum channel $\Phi: \mathcal{B}(\mathcal{H})\rightarrow \mathcal{B}(\mathcal{K})$ with its Choi-Jamio{\l}kowski image $J(\Phi)$ with the normalised eigensystem $\left\lbrace \lambda_m, |x_m\>\right\rbrace $. Then its Kraus operators are given by the following expression:
	\be\label{194}
	K_m=\sqrt{\lambda_m}X_m^{T},
	\ee
	where $X_m=\cV^{-1}(|x_m\>)$ is the matrix obtained by the inverse vectorization of the eigenvector $|x_m\>$.
\end{lemma}

\begin{proof}
	At the first step let us calculate an action of the channel $\Phi$ on some state $\rho$:
	\be
	\label{63}
	\begin{split}
		\Phi(\rho)&=\sum_{kl}\left(\Phi(\rho)\right)_{kl}E_{kl}=\sum_{kl}\tr\left[E_{kl}^{\dagger}\Phi(\rho) \right]E_{kl}=\sum_{kl}\tr\left[E_{lk}\Phi(\rho) \right]E_{kl} =\\ 
		&=\sum_{kl}\tr\left[E_{lk}\Phi\left(\sum_{ij}\tr\left(E_{ji}\rho\right)E_{ij}  \right) \right]E_{kl}=\sum_{ijkl}\tr\left[E_{ji}\rho \right]\tr\left[E_{lk}\Phi(E_{ij})\right]E_{kl}.
	\end{split}
	\ee
	Notice that $\<k|\Phi(E_{ij})|l\>=J(\Phi)_{ik,jl}$, indeed we have
	\be
	J(\Phi)_{ik,jl}=\<ik|J(\Phi)|jl\>=\<ik|\sum_{rs}E_{rs}\otimes \Phi(E_{rs})|jl\>=\<k|\Phi(E_{ij})|l\>.
	\ee
	We know that $J(\Phi)\geq 0$. Thanks to this we can write the spectral decomposition of $J(\Phi)$ as $J(\Phi)=\sum_m \lambda_m|x_m\>\<x_m|$, where $\lambda_m$ are the eigenvalues of $J(\Phi)$ and $\left\lbrace |x_m\> \right\rbrace $ is the set of its normalised eigenvectors. Now using the above mentioned decomposition let us compute the matrix element $\left(J(\Phi) \right)_{ik,jl} $ in the operator basis $\left\lbrace E_{ij}\right\rbrace$:
	\be
	\label{65}
	\left(J(\Phi) \right)_{ik,jl}=\sum_m \lambda_m\<ik|x_m\>\<x_m|jl\>.
	\ee
	Finally putting~\cref{65} into~\cref{63} we have:
	\be
	\label{km}
	\begin{split}
		\Phi(\rho)&=\sum_{ijkl}\sum_m \lambda_m\<ik|x_m\>\<x_m|jl\>\<i|\rho|j\>|k\>\<l|=\sum_{ijkl}\sum_m\lambda_m \<ik|x_m\>\<x_m|jl\>|k\>\<i|\rho|j\>\<l|\\
		&=\sum_{ijkl}\sum_m\lambda_m\<ik|x_m\>\<x_m|jl\>E_{ki}\rho E_{jl}=\sum_m\left(\sum_{ki}\sqrt{\lambda_m}\<ik|x_m\>E_{ki} \right)\rho\left(\sum_{jl}\sqrt{\lambda_m}\<x_m|jl\>E_{jl}\right)\\
		&= \sum_mK_m\rho K_m^{\dagger},
	\end{split}
	\ee
where
\begin{align}
K_m &:= \sqrt{\lambda_m}\sum_{ki}\<ik|x_m\>E_{ki}.
\end{align}
It is easy to see that
	\be
	\sum_{ki}\<ik|x_m\>E_{ki}=\left[\cV^{-1}(|x_m\>) \right]^{T}=X_m^{T},
	\ee
	where $\cV^{-1}(\cdot)$ is the inverse vectorisation, and hence \cref{194} holds.

We establish below that $\{K_m\}$ is indeed a set of Kraus operators for the ICQC $\Phi$ by showing that $\sum_m K_m^{\dagger}K_m=\text{\noindent
		\(\mathds{1}\)}$. Indeed we have
	\be
	\begin{split}
		&\sum_m K_m^{\dagger}K_m=\sum_m \lambda_m \sum_{ijkl}\<x_m|ik\>E_{ik}\<jl|x_m\>E_{lj}=\sum_m\sum_{ijk}\<jk|x_m\>\<x_m|ik\>E_{ij}=\\
		&=\sum_{ijk}\sum_{st}\<jk|E_{st}\otimes \Phi\left(E_{st} \right)|ik\>E_{ij}=\sum_{ijk}\<k|\Phi \left(E_{ji} \right)|k\>E_{ij}=\sum_{ij}\tr\left[\Phi \left( E_{ji}\right)  \right] E_{ij}.
	\end{split}
	\ee
	Since $\Phi$ is a CPTP map, we have $\tr\left[\Phi \left( E_{ji}\right)  \right]=\delta_{ji}$, since $\tr \left(E_{ji} \right)=\delta_{ji}$. This completes the proof.
\end{proof}
\newpage
\bibliographystyle{plain}
\bibliography{References3}
\end{document}